\documentclass[a4paper, 11pt]{article}

\usepackage{amsthm}
\usepackage{amsmath}
\usepackage{amssymb}
\usepackage{graphicx}
\usepackage[round]{natbib}
\usepackage{bbm}
\usepackage{textcomp}
\usepackage{booktabs}
\usepackage{url}
\usepackage[margin=2.5cm]{geometry}
\usepackage[colorlinks,citecolor=blue]{hyperref}

\theoremstyle{plain}  
\newtheorem{thm}{Theorem}[section] 
 
\newtheorem{prop}[thm]{Proposition}

\theoremstyle{definition} 
\newtheorem{defn}{Definition}[section] 
 
\newtheorem{exmp}{Example}[section]

\theoremstyle{remark} 
\newtheorem*{rem}{Remark}

\newcommand{\diff}{\,\mathrm{d}}
\newcommand{\E}{\mathbb{E}}

\newcommand{\R}{\mathbb{R}}
\newcommand{\N}{\mathbb{N}}

\newcommand{\F}{\mathcal{F}}

\newcommand{\one}{\mathbbm{1}}

\newcommand{\leST}{\le_{\mathrm{st}}}
\newcommand{\geST}{\ge_{\mathrm{st}}}

\DeclareMathOperator*{\argmax}{arg\,max}

\newcommand{\beginsupplement}{%
	\setcounter{table}{0}
	\renewcommand{\thetable}{S\arabic{table}}%
	\setcounter{figure}{0}
	\renewcommand{\thefigure}{S\arabic{figure}}%
}

\begin{document}

\title{Sequentially valid tests for forecast calibration}
\author{Sebastian Arnold \and Alexander Henzi \and Johanna F.~Ziegel \and \\ \texttt{\{sebastian.arnold,alexander.henzi,johanna.ziegel\}@stat.unibe.ch}}
\date{November 8, 2021}
\maketitle

\begin{abstract}
Forecasting and forecast evaluation are inherently sequential tasks. Predictions are often issued on a regular basis, such as every hour, day, or month, and their quality is monitored continuously. However, the classical statistical tools for forecast evaluation are static, in the sense that statistical tests for forecast calibration are only valid if the evaluation period is fixed in advance. Recently, e-values have been introduced as a new, dynamic method for assessing statistical significance. An e-value is a non-negative random variable with expected value at most one under a null hypothesis. Large e-values give evidence against the null hypothesis, and the multiplicative inverse of an e-value is a conservative p-value. E-values are particularly suitable for sequential forecast evaluation, since they naturally lead to statistical tests which are valid under optional stopping. This article proposes e-values for testing probabilistic calibration of forecasts, which is one of the most important notions of calibration. The proposed methods are also more generally applicable for sequential goodness-of-fit testing. We demonstrate in a simulation study that the e-values are competitive in terms of power when compared to extant methods, which do not allow for sequential testing. In this context, we introduce test power heat matrices, a graphical tool to compactly visualize results of simulation studies on test power. In a case study, we show that the e-values provide important and new useful insights in the evaluation of probabilistic weather forecasts.
\end{abstract}

\section{Introduction}
Probabilistic forecasts incorporate the uncertainty about a future quantity $Y$ comprehensively in the form of probability distributions. A minimal requirement for useful probabilistic forecasts is calibration, meaning that the predicted probabilities should conform with the actual observed event frequencies. This article develops novel statistical tools to validate \emph{probabilistic calibration}, one of the most prominent and widely applied notions of calibration. Probabilistic calibration requires that $Y$ should be below the $\alpha$-quantile of the forecast distribution with a frequency of about $\alpha \cdot 100\%$, for all $\alpha \in (0,1)$. More precisely, the predictive cumulative distribution function (CDF) $F$ is evaluated at the outcome $Y$, and this quantity, suitably randomized in case of discontinuities of $F$, is called the {probability integral transform (PIT)} and should be uniformly distributed on  $(0,1)$ for a probabilistically calibrated forecast. Checks of the uniformity of the PIT, and of the closely related rank histogram, constitute a cornerstone of forecast evaluation \citep{Dawid1984,DieboldGuntherETAL1998,Hamill2001,Gneiting2007}.

From a statistical point of view, testing probabilistic calibration for forecasts with lag $1$ is straightforward. For example, in the case of daily forecasts issued for the next day, it suffices to apply any goodness of fit test for the standard uniform distribution to a sample of the PIT from a series of forecasts and observations. However, it has been noted early on that statistical tests alone are not informative enough, because they do not indicate the type of misspecification \citep{DieboldGuntherETAL1998}. Therefore, tests of calibration are commonly accompanied by a histogram plot of the PIT distribution, which allows to identify classical types of misspecification at a glance, namely, biased forecasts lead to PIT histograms skewed to the left or the right, and under- or overdispersed forecasts yield U-shaped or inverse U-shaped PIT histograms, respectively. 

We argue that a drawback of the established tools for validating probabilistic calibration is that they do not fully account for the sequential nature of forecasting. The relationship between forecasts and observations is often complicated and forecast misspecification changes over time. However, the classical tools for validating calibration require the sample size to be fixed in advance and independently of the data. As an illustrative example, consider a weather forecaster who, after updating a prediction model, monitors the quality of daily forecasts and wants to check if the new forecasts are probabilistically calibrated. She aims at a sample size of one year, and plans to check uniformity of the PIT at the end of the observation period. If, by chance, the forecaster realizes after half of observation period that the forecast is strongly biased, then a p-value from a classical goodness of fit test with all data at this time point is not valid since the sample size depends on the data. On the other hand, if the forecaster does not look at the data until the end of the observation period, then the PIT distribution with the sample of the full year could again be close to uniform, for example if there is a change in the direction of the bias in the second half of the year, and the forecaster is unable to detect the misspecification. Such effects appear in practice as exemplified in the case study in Section \ref{sec:case_study}. Any analysis of sub-periods has to be planned in advance, which is often difficult and cumbersome since it is usually not known in advance how forecast misspecification changes over time and what discretizations of the time domain are appropriate.

This article develops new methodology for checks of probabilistic and related notions of calibration in a sequential setting, based on the new concept of \emph{e-values}. E-values have received an increasing interest in recent years, see \citet{Vovk2021}, \citet{GrunwaldHeideETAL2019}, \citet{Shafer2021} (who uses the term \emph{betting score}), \citet{Ramdas2020}, and the references therein. An e-value is a non-negative random variable $E$ such that for all distributions $P$ in a set $\mathcal{P}$, the null hypothesis, the inequality $\mathbb{E}_{P}(E) \leq 1$ holds. E-values can be easily transformed into (conservative) p-values since $P(1/E \leq \alpha) \leq \alpha$ for $\alpha \in (0,1)$ by Markov's inequality, and large e-values give evidence against the null hypothesis. One main motivation for using e-values is their simple behaviour under combinations. Convex combinations of e-values are again e-values, and so is the product of independent e-values. In a sequential setting, if $(E_t)_{t \in \mathbb{N}}$ is a sequence of e-values adapted to a filtration $(\mathcal{F}_t)_{t \in \mathbb{N}}$, then by the tower property of conditional expectations the process $e_t = \prod_{i = 1}^t E_i$, $t \in \mathbb{N}$, is a non-negative supermartingale or \emph{test martingale} \citep{Shafer2011} and it satisfies Ville's inequality, that is $P(\sup_{t \in \mathbb{N}} e_t \geq 1/\alpha) \leq \alpha$; see \citet{Ramdas2020} for a comprehensive analysis of non-negative martingales for statistical testing. In the example of the weather forecaster from the previous paragraph, this implies that with e-values the forecaster may reject the hypothesis of calibration at the level $\alpha$ as soon as the process $(e_t)_{t \in \mathbb{N}}$ exceeds $1/\alpha$, without having to fix a sample size in advance. The forecaster is allowed to monitor the PIT and the process $(e_t)_{t \in \mathbb{N}}$ in real time. In the special case of a simple null hypothesis, $\mathcal{P} = \{P_0\}$, e-values take the form of likelihood ratios or Bayes factors \citep{GrunwaldHeideETAL2019}. In particular, e-values for testing the null hypothesis that a quantity $Z \in (0,1)$ is uniformly distributed on the unit interval, short $\mathrm{UNIF}(0,1)$, are Lebesgue densities on $[0,1]$. It is therefore simple to construct valid e-values for testing probabilistic calibration or, detached from the forecasting context, goodness-of-fit testing of the $\mathrm{UNIF}(0,1)$ distribution, with ``valid'' referring to type one error guarantees. The non-trivial task in the construction of e-values is to achieve sufficient power to detect violations of the null hypothesis. \citet{Shafer2021} calls strategies for constructing e-values ``betting strategies'', since an e-value can be interpreted as a bet against the null hypothesis and the process $(e_t)_{t \in \mathbb{N}}$ corresponds to the capital over time if all -- or part of the -- gains are reinvested into the new bet at each $t \in \mathbb{N}$.

The contributions of this article are as follows. A brief review of related literature is given in Section \ref{sec:literature}. In Section \ref{sec:e_values} we construct e-values for testing the null hypothesis that a quantity $Z \in [0,1]$ is distributed according to $\mathrm{UNIF}(0,1)$, and for the analogous hypothesis that a discrete $R \in \{1, \dots, m\}$ follows a uniform distribution on the integers $1$ to $m$, short $\mathrm{UNIF}(\{1, \dots, m\})$.  These hypotheses appear naturally in calibration checks for probabilistic forecasts and ensemble forecasts, and precise definitions of forecast calibration are given in Section \ref{sec:preliminaries}. Furthermore, we characterize and construct e-values for the weaker hypotheses that a random variable $Z \in [0,1]$ with distribution $P$ is stochastically smaller than $\mathrm{UNIF}(0,1)$, short $P \leST \mathrm{UNIF}(0,1)$, which means $P(Z \leq z) \geq z$ for all $z \in (0,1)$. This hypothesis appears in a new definition of calibration for quantile forecasts which is closely related to usual probabilistic calibration. Section \ref{sec:e_values} is of interest independent from the forecasting context, and the methods can also be applied to general goodness-of-fit or stochastic order testing problems in sequential settings. Proofs of theoretical results are deferred to Appendix \ref{app:proofs}. In Section \ref{sec:simulation_study} we demonstrate that the e-values are competitive in terms of power when compared to established tests. Here, we suggest a new graphical tool to compactly display simulation results on test power, so-called test power heat matrices. Section \ref{sec:case_study} presents an application to testing calibration of postprocessed weather forecasts, and we show that the e-values give rise to novel and informative graphical tools for the sequential evaluation of forecast calibration.

\section{Related literature} \label{sec:literature}
Testing based on e-values and nonnegative supermartingales has received increasing interest in recent years, but many of the underlying concepts have been present in the literature much earlier. A historical review and comparisons to other concepts in the literature are given by \citet[Section 7]{GrunwaldHeideETAL2019}; the aim of this section is to draw connections to works closely related to the specific testing problems and methods considered in this article.

The earliest and still the most prominent method for sequential testing is Wald's Sequential Probability Ratio Test \citep[SPRT;][]{Wald1945, Wald1947}. In its simplest form, the SPRT tests whether a sequence $(Y_n)_{n \in \mathbb{N}}$ of independent random variables follow a density $f_0$, a simple null hypothesis. This can be achieved with the likelihood ratio process $e_t = \prod_{i=1}^{t}f_i(Y_i)/f_0(Y_i)$, which is a nonnegative supermartingale under the null hypothesis if the alternative densities $f_i$ are predictable, i.e.~depend only on $Y_1, \dots, Y_{i-1}$ for $i \in \mathbb{N}$. Two hypotheses considered in this article, the continuous and discrete uniform distribution (Sections \ref{subsec:continuous}, \ref{subsec:discrete}), are concerned with simple null hypotheses and could hence also be seen as an application of the SPRT. As discussed by \citet[Section 7]{GrunwaldHeideETAL2019}, the difference between the SPRT and testing based on e-values in the case of a simple null hypothesis is more a conceptual one, in that Wald proposed to stop testing and reject or accept the null hypothesis as soon as the likelihood ratio process crosses pre-specified boundaries. For the sequential evaluation and monitoring of forecasts, such rigid stopping criteria are usually not desirable.

Various tests for forecast calibration and, in particular, probabilistic calibration are available in the literature, see for example the references in \citet[Section 3.1]{Gneiting2007}, but these are non-sequential tests. There are some sequential testing methods which are applicable, but not necessarily tailored to the problems considered in this article, in particular, to calibration testing of weather forecasts. \citet{Howard2022} develop confidence sequences for cumulative distribution functions and quantiles, which could be used both for testing continuous and discrete uniformity and stochastic dominance. However, the test for stochastic dominance requires independent and identically distributed (iid) observations, whereas our hypothesis (Definition 3.4 (iii) in the next Section) allows dependence and time-varying distributions. Furthermore, there are sequential tests for exchangeability and the iid assumption, reviewed by \citet{Vovk2021a}, which are based on testing whether a sequence of so-called conformal p-values is independent and uniformly distributed, like the PIT of a probabilistically calibrated one step ahead forecast. While these methods test the same null hypothesis, the types of misspecification are often different from the ones in forecast evaluation.

\citet{HenziZiegel2021} and \citet{Choe2021} give a first application of e-values and related concepts to testing probability forecast superiority. Their articles are concerned with comparing probability predictions $p_t$, $q_t \in [0,1]$ for a binary event $Y_{t+h}\in \{0,1\}$ with respect to so-called proper scoring rules $S$, such as the squared error $S(p, y) = (p-y)^2$. In such comparisons, forecast superiority depends on both, the calibration of the forecasts, i.e.~whether the predicted probabilities $p_t$, $q_t$ match the observed event frequencies, and on their sharpness, with sharper forecasts being probability predictions close to $0$ or $1$ rather than to $0.5$. In contrast, this article develops tests for calibration only, which do not give an indication of the relative performance of competing methods.

\cite{Gneiting2021} provide a comprehensive overview of the available calibration notions for real-valued outcomes, which are traditionally formulated for one-period  predictions for an outcome $Y$. Possible extensions from the one-period setting to sequential notions of calibration can be found in \cite{Straehl2017}. In Definition \ref{def:time_series_calibration}, we suggest a natural sequential version of probabilistic calibration that differs slightly from \citet[Definition 2.7]{Straehl2017}. 

For multivariate outcomes, generalizations of probabilistic calibration have been proposed by \citet{Thorarinsdottir_et_al_2016, Ziegel_Gneiting_2014}. Checking for calibration of multivariate probabilistic predictions with respect to these notions of calibration results again in assessing whether a certain univariate statistic is uniformly distributed. Therefore, the methods suggested in this article directly carry over to this scenario. 


\section{Probabilistic calibration}\label{sec:preliminaries}
Let $Y$ be a real-valued outcome defined on a probability space $(\Omega,\F, P)$. We denote by $F$ the CDF associated with a (random) probabilistic forecast for $Y$. Definitions \ref{def:pit} and \ref{def:rankhist} below are standard notions of forecast calibration, see for example \citet{Gneiting2007, Gneiting2013}.

\begin{defn} \label{def:pit}
The \emph{probability integral transform (PIT)} of a forecast $F$ for an outcome $Y$ is defined as $Z_F(Y) = F(Y-) + V(F(Y) - F(Y-))$, where $F(y-) = \lim_{z \rightarrow y, z < y} F(z)$ and $V$ is a uniformly distributed random variable on $(0,1)$ independent of the pair $(F, Y)$.  The forecast $F$ is \emph{probabilistically calibrated} if $Z_F(Y)  \sim \mathrm{UNIF}(0,1)$.
\end{defn}

Of great importance in weather forecasting are ensemble forecasts \citep{Bauer2015}. An ensemble forecast is a collection of point forecasts generated by running a numerical weather prediction (NWP) model $m$ times, typically $m = 20$ to $50$, each time with different initial conditions, which allows to quantify the forecast uncertainty. We denote ensemble forecasts by vectors $\boldsymbol{X}=(X_1, \dots, X_m) \in \mathbb{R}^m$. To define calibration, let the (randomized) rank of $Y$ equal
\begin{equation} \label{eq:rank}
	\mathrm{rank}_{\boldsymbol{X}}(Y)=1 + \#\big\{i=1, \dots, m \mid X_i < Y\big\} + W \ \in \{1, \dots, m+1\},
\end{equation}
where $W$ is a random variable that equals zero almost surely if $N=\#\big\{i =1, \dots, m \mid X_i=Y\big\}$ is zero, and is uniformly distributed on $\big\{1, \dots, N\}$ otherwise.

\begin{defn} \label{def:rankhist}
An ensemble forecast $\boldsymbol{X}$ is \emph{rank calibrated} if $\mathrm{rank}_{\boldsymbol{X}}(Y) \sim \mathrm{UNIF}\big(\{1, \dots, m+1\}\big)$.
\end{defn}

\begin{rem}
Rank calibration is commonly assessed with the rank histogram, a plot of the empirical frequencies of the ranks over a sample \citep{Anderson1996}. We use a randomization of the rank in case of ties because with this convention the PIT and the rank of $Y$ are related via the equation $\mathrm{rank}_{\boldsymbol{X}}(Y) = 1 + \lfloor mZ_{F_{\boldsymbol{X}}}(Y) \rfloor$, where $F_{\boldsymbol{X}}$ is the empirical CDF (ECDF) of the ensemble $\boldsymbol{X}$.  The definition of the rank given in \eqref{eq:rank} slightly generalizes the unified PIT introduced by \citet[][p.~374]{Vogel2018} for evaluating precipitation forecasts, which randomizes ranks in case of multiple occurrences of zero forecasts.
\end{rem}

We proceed to introduce a new, closely related notion of calibration for quantile forecasts. Let $\alpha_0=0< \alpha_1 < \dots < \alpha_K < 1=\alpha_{K+1}$ be $K$ quantile levels. Instead of issuing a complete predictive CDF for the unknown quantity $Y$, we only aim to give point forecasts $q_1 \leq \dots \leq q_K$ for the quantiles of the distribution of $Y$ at levels $\alpha_1,\dots,\alpha_K$. Recall, that $q_i$ is an $\alpha_i$-quantile of $F$ if
\[
	F(q_i -) \le \alpha_i \le F(q_i), \quad i=1,\dots,K.
\]
Therefore the set of quantile forecasts can be interpreted as a partial disclosure of the predictive CDF $F$. With $q_0=-\infty$ and $q_{K+1}=\infty$, define 
\[
	F_u(y) := \sum_{i=1}^{K+1} (\alpha_i-\alpha_{i-1}) \one\{q_{i} \le y\}, \quad F_\ell(y) := \sum_{i=0}^{K} (\alpha_{i+1}-\alpha_{i}) \one\{q_i \le y\}, \quad y \in \mathbb{R}.
\] 

\begin{prop} \label{prop:qpit}
Let $0<\alpha_1 < \dots < \alpha_K < 1$ be $K$ quantile levels. Any (deterministic) CDF $F$ with corresponding quantiles $q_1 \leq \dots \leq q_K$ satisfies 
\begin{equation}\label{eq:stochorder}
	F_u(y) \le F(y) \le F_\ell(y), \quad y \in \mathbb{R}.
\end{equation} 
Furthermore 
\begin{equation}\label{eq:quantiles}
	F_{u}(q_i -) \le \alpha_i \le F_u(q_i) \quad \textrm{and} \quad F_{\ell}(q_i -) \le \alpha_i \le F_\ell(q_i), \quad i=1,\dots,K.
\end{equation}
\end{prop}

Similar to the classical PIT, we define
\[
	Z_{F_u}(Y) := V F_u(Y) + (1-V) F_u(Y-) \hspace{0.2cm}\textrm{and}\hspace{0.2cm} Z_{F_\ell}(Y) := V F_\ell(Y) + (1-V) F_\ell(Y-),
\]
where $V$ is a uniformly distributed random variable on $(0,1)$ independent of $Y$ and the quantile predictions $q_1, \dots, q_K.$ In the sequel, $Z_{F_u}(Y)$ and $Z_{F_l}(Y)$ will be referred to as the \emph{upper} and \emph{lower quantile PIT}. By \eqref{eq:stochorder} these quantities satisfy
\[
	Z_{F_u}(Y) \le Z_F(Y) \le Z_{F_\ell}(Y)
\]
almost surely, and $Z_{F_\ell}(Y)-Z_{F_u}(Y) \le \sup_{i=0, \dots, K}(\alpha_{i+1}-\alpha_i)$ with equality if $\alpha_{i+1} - \alpha_i = c > 0$ for all $i$. In this case, which is the important special case of equispaced quantile levels, $Z_{F_\ell}(Y) = Z_{F_u}(Y) + c$. 

\begin{defn} \label{def:calibration_quantile}
A set of quantile forecasts $q_1 < \dots < q_K$ is \emph{probabilistically calibrated} if
\[
	Z_{F_u}(Y) \leST \mathrm{UNIF}(0,1) \leST Z_{F_\ell}(Y).
\]
\end{defn}

\begin{rem}
The functions $F_u$ and $F_\ell$ are defective CDFs in the sense that $\lim_{y \to \infty} F_u(y) = \alpha_K < 1$ and $\lim_{y \to -\infty} F_{\ell}(y) = \alpha_1 > 0$, respectively, but they satisfy the remaining conditions for being a CDF. Note that the distribution of the upper (lower) quantile PIT is stochastically smaller (greater) than $\mathrm{UNIF}(0,1)$, which implies that its CDF is pointwise greater (smaller) than the uniform CDF for probabilistically calibrated quantile forecasts.
\end{rem}

The definitions of calibration introduced so far do not include any notion of time. In practice, for example for the PIT, one observes a time series $(F_t, Y_t)_{t \in \mathbb{N}}$ of forecasts and observations, where $F_t$ is the forecast for a lagged observation $Y_{t + h}$, with a fixed integer lag $h \geq 1$. The definition below formalizes calibration for such sequential settings.

\begin{defn} \label{def:time_series_calibration}
Let $(\Omega, \mathcal{F},P)$ be a probability space with a filtration $(\mathcal{F}_t)_{t \in \N}$, and $h$ be a positive integer. Let further $(Y_t)_{t \in \N}$ be an adapted sequence of observations.
\begin{itemize}
\item[(i)] A sequence of probability forecasts $(F_{t})_{t \in \N}$ is \emph{probabilistically calibrated at lag $h$} if
\[
	\mathcal{L}(Z_{F_t}(Y_{t+h}) \mid Z_{F_j}(Y_{j + h}), 0\leq j \leq t - h) = \mathrm{UNIF}(0,1), \ t \in \mathbb{N}.
\]
\item[(ii)] A sequence of ensemble forecasts $(\boldsymbol{X}_t)_{t \in \N}$ of size $m$ is \emph{rank calibrated at lag $h$} if
\[
	\mathcal{L}(\mathrm{rank}_{\boldsymbol{X}_t}(Y_{t+h}) \mid \mathrm{rank}_{\boldsymbol{X}_j}(Y_{j + h}), 0 \leq j \leq t - h) = \mathrm{UNIF}(\{1, \dots, m + 1\}), \ t \in \mathbb{N}.
\]
\item[(iii)] A sequence of quantile forecasts $(q_{1;t}, \dots, q_{K;t})_{t \in \N}$ is \emph{probabilistically calibrated at lag $h$} if
\[
	Z_{F_{u;t}}(Y_{t+h}) \leST \mathrm{UNIF}(0,1) \leST Z_{F_{\ell;t}}(Y_{t+h})
\]
conditional on $Z_{F_{u;j}}(Y_{j+h}), Z_{F_{\ell;j}}(Y_{j+h})$, $0 \leq j \leq t - h$, for $t \in \mathbb{N}$.
\end{itemize} 
\end{defn}
Note that for $t\leq h$ there is no conditioning in all cases, and the requirements in (i)-(iii) are understood to hold unconditionally. Furthermore we silently assume existence of a sequence $(V_t)_{t \in \N}$ of adapted, independent $\mathrm{UNIF}(0,1)$ variables defined on the probability space, independent of all other objects, in parts (i) and (iii) of Definition \ref{def:time_series_calibration} to define the PIT and quantile PIT. Similarly, existence of an analogous sequence $(W_t)_{t \in \N}$ for the randomization of the ranks in part (ii) is assumed. For forecasts with lag $h > 1$, the definition does not condition on $Z_{F_j}(Y_{j+h})$ with $t - h < j < t$, since the corresponding observations are not yet available at time $t$ when the forecasts are issued, and in this case, the joint distribution of the PIT (or ranks, quantile PIT) which are less than $t$ time units apart is not specified.

\begin{rem}
With lag $h = 1$, the definition of calibration implies that the sequence of the PIT, $(Z_{F_t}(Y_{t+1}))_{t \in \mathbb{N}}$, or of the ranks, $(\mathrm{rank}_{\boldsymbol{X}_t}(Y_{t+1}))_{t \in \mathbb{N}}$, are independent, since the conditional distributions in part (i) or (ii) do not depend on the past values in the sequence. Indeed, for a probabilistically calibrated sequence of forecasts $ (F_{t})_{t \in \N}$ and $v,w \in [0,1]$, it holds
\[
	P(Z_{F_1}(Y_2)\leq v, Z_{F_2}(Y_3)\leq w) = P( Z_{F_2}(Y_3)\leq w \mid Z_{F_1}(Y_2) \leq v) P(Z_{F_1}(Y_2)\leq v)= v w,
\]
and it follows inductively that $(Z_{F_t}(Y_{t + 1}))_{t \in \mathbb{N}}$ are independent. Hence the definition of probabilistic calibration corresponds to the classic definition given in  \citet{DieboldGuntherETAL1998}. However, for lag $h > 1$ and for the quantile PIT, where no particular conditional distribution is assumed, there may be dependence in the sequence of PITs, ranks, or quantile PITs. 
\end{rem}

\begin{exmp}
This example is a sequential adaption of Example 3 in \citet{Gneiting2007}, see also \citet[Section 6]{Tsyplakov2011}. Consider a sequence of random variables $(Y_t)_{t\in \N}$, where $Y_{t+1} \sim \mathcal{N}(Y_t,1)$ and $Y_1 \sim \mathcal{N}(0,1)$, and the sequence of \emph{unfocused forecasts} 
\begin{equation*}
	F_t(y) = \frac{1}{2}\big(\Phi(y-Y_t)+\Phi(y-Y_t-\eta_t)\big),
\end{equation*}
where $(\eta_t)_{t \in \N}$ is an independent sequence for which $\eta_t$ attains the values $\pm 1$ with equal probability for each $t$. Since $F_t(Y_{t+1})$ is distributed as $\frac{1}{2}\big(\Phi(Z_t)+\Phi(Z_t-\eta_t)\big)$ for a sequence $(Z_n)_{n \in \N}$ of iid standard normally distributed $Z_n$, it is analogous as in the non-sequential setting to show that this sequence of forecasts is probabilistically calibrated at lag 1 in the sense of Definition \ref{def:time_series_calibration}, see \citet[Example 3]{Gneiting2007}. However, the predictive distributions are not correct in the sense that $\mathcal{L}(Y_{t+1}\mid F_t)= F_t$ does not hold almost surely, which is a stronger property called autocalibration \citep{Gneiting2013}. This demonstrates that in Definition \ref{def:time_series_calibration}, it is crucial that the conditioning only involves the past values of the PIT, ranks, or quantile PIT, as otherwise different notions of calibration would be obtained.
\end{exmp}

\begin{rem}
The above example illustrates that autocalibration and our sequential definition of probabilistic calibration are not equivalent in general. \citet[Theorem S3.1]{Gneiting2021} show that for $h=1$ these two notions of calibrations coincide, if and only if, the forecast $F_t$ only uses information of the past observations $Y_1, \dots, Y_{t-1}$, that is, $F_t$ is measurable with respect to $\sigma(Y_1, \dots, Y_{t-1})$. In practice this is usually not the case, as predictions do not only depend on the past time series of observations but also include side information, like expertise or predictions from numerical weather prediction models, and hence a sequence of probabilistically calibrated forecasts is not necessarily autocalibrated.
\end{rem}

\section{E-values} \label{sec:e_values}

\subsection{E-values in sequential settings}
We proceed with formal definitions and properties of e-values in sequential settings. The notation largely follows \citet{Vovk2021}, but we formalize a new concept of lagged sequential e-values which is particularly relevant in forecast evaluation. Throughout this section, let $(\Omega, \mathcal{F})$ be the underlying measurable space and $\mathcal{P}$ be a suitable set of distributions.

\begin{defn}
Let $\mathcal{H}, \mathcal{H}' \subset \mathcal{P}$. An \emph{e-value for $\mathcal{H}$} is a non-negative random variable $E$ such that $\E_P E \le 1$ for all $P \in \mathcal{H}$. An e-value for $\mathcal{H}$ is \emph{testing $\mathcal{H}$ against $\mathcal{H}'$} if $\E_Q E > 1$ for all $Q\in\mathcal{H}'$.
\end{defn} 

\begin{defn} \label{def:sequential_e}
Let $(\mathcal{F}_t)_{t \in \N}$ be a filtration, $h$ be a positive integer and $\mathcal{H}, \mathcal{H}' \subset \mathcal{P}$. Adapted non-negative random variables $(E_t)_{t \in \N}$ are called \emph{sequential e-values for $\mathcal{H}$ at lag $h$} if $\E_P(E_t \mid \F_{t-h})\leq 1$ for all $P \in \mathcal{H}$ and for all $t\in \N$. Sequential e-values for $\mathcal{H}$ at lag $h$ are \emph{testing $\mathcal{H}$ against $\mathcal{H}'$} if $\E_Q (E_t \mid \F_{t-h}) > 1$ for all $Q\in\mathcal{H}'$ and for all $t \in \N$. For $t\leq h$ expectations are understood unconditionally.  
\end{defn}

A motivation for Definition \ref{def:sequential_e} is the sequential evaluation of forecasts with prediction horizon $h$. At time $t$, we observe the current quantity of interest, $Y_t$, and the forecast $F_t$ for $Y_{t + h}$. We are interested in testing the null hypothesis that the forecasts are calibrated, for example in the sense of Definition \ref{def:time_series_calibration}. Forecast evaluation is normally based on the observation and on information available at the time of forecasting. Therefore, an e-value $E_t$ for testing calibration at time $t \geq h + 1$ should satisfy $\mathbb{E}(E_t \mid \mathcal{F}_{t - h}) \leq 1$, where $(\mathcal{F}_t)_{t \in \mathbb{N}}$ is a suitable filtration. But $\mathbb{E}(E_t \mid \mathcal{F}_{j}) \leq 1$ may be violated for $t-h < j < h$ even for calibrated forecasts, since the conditional expectation involves information not available at the time of forecasting. Therefore, e-values $(E_t)_{t \in \mathbb{N}}$ for testing forecast calibration should be sequential at lag $h$.

For lag $h = 1$, sequential e-values can be combined by their cumulative product. For $h > 1$ we can combine e-values with a U-statistics approach (see \citet{Vovk2021} and \citet[Proposition 3.4]{HenziZiegel2021}).

\begin{prop} \label{prop:u_statistics} Let $(E_t)_{t \in \N}$ be sequential e-values for $\mathcal{H}\subset \mathcal{P}$ at lag $h$ adapted to the filtration $(\mathcal{F}_t)_{t \in \N}$. Then for all $T\geq h+1$, with $I_k(T)=\{k+hs: s= 0,\dots, \lfloor (T-k)/h\rfloor\}$,
\begin{equation} \label{eq:u_statistics}
	e_T = \frac{1}{h} \sum_{k=1}^{h} \prod_{l \in I_k(T)} E_l
\end{equation}
is $\F_T$ measurable and an e-value for $\mathcal{H}$. For any stopping time $\tau$, the process $(e_t)_{t \in \mathbb{N}}$ satisfies
\begin{equation} \label{eq:stop_lag_h}
	\mathbb{E}_{P}(e_{\tau+h-1}) \leq 1, \ P \in \mathcal{H}.
\end{equation}
\end{prop}

In the case $h = 1$, $(e_t)_{t \in \mathbb{N}}$ is a non-negative supermartingale. This implies that $\mathbb{E}_P(e_{\tau}) \leq 1$ for all $P \in \mathcal{H}$ and all stopping times $\tau$, which is the characterizing property of \emph{e-processes} \citep{Ramdas2022eprocess} and a strictly weaker property than being a non-negative supermartingale. This allows valid inference at arbitrary, not even necessarily pre-specified stopping times. For rejecting the null hypothesis at a fixed level $\alpha \in (0,1)$, one may apply the aggressive stopping criterion
\[
	\tau_{\alpha} = \inf\{t \in \mathbb{N} : e_t \geq 1/\alpha\}.
\]
In case when $h > 1$, the process $(e_t)_{t \in \mathbb{N}}$ is in general not an e-process and hence also not a non-negative supermartingale. The reason for this is that for $k = 1, \dots, h$, the sub-process $M^{[k]} =  (\prod_{l \in I_k(t)} E_l)_{t \in \mathbb{N}}$ is a supermartingale only with respect to the filtration $\mathfrak{F}^{[k]} = (\mathcal{F}_{\lfloor (t-k)/h\rfloor h+k})_{t \in \mathbb{N}}$, but not with respect to $\mathfrak{F} = (\mathcal{F}_t)_{t \in \N}$, and averages of supermartingales with respect to different filtrations are not supermartingales in general. It may hold that for some $t_0 \geq h$ and $j \in \{1, \dots, h - 1\}$, the e-value $E_{t_0}$ satisfies $\mathbb{E}_P(E_{t_0}\mid \mathcal{F}_{t_0-j}) > 1$, which allows to find $\mathfrak{F}$-stopping times $\tau$ such that $\mathbb{E}_P(e_{\tau}) > 1$. In practice, \eqref{eq:stop_lag_h} means that when deciding to stop testing at a time $t_0$, this decision must not involve information available after time $t_0-h+1$, or otherwise one also has to include the future e-values $E_{t_0+1}, \dots, E_{t_0+h-1}$ in the evaluation. However, as shown in Appendix \ref{app:proofs}, rescaling the process $e_T$ by $(e\log(h))^{-1} < 1$, with $e = \exp(1)$, still allows to construct threshold tests at level $\alpha \in (0,1)$. More precisely, for $h > 1$ the stopping time
\[
	\tau_{\alpha,h} = \inf\left\{t \in \mathbb{N} : \frac{1}{he\log(h)} \sum_{k=1}^{h} \sup_{s \leq t} \prod_{l \in I_k(s)} E_l \geq 1/\alpha\right\}
\]
guarantees that $P(\tau_{\alpha,h} < \infty) \leq \alpha$ for all $P \in \mathcal{H}$. If, as in this article, the e-values $E_t$ are of the form $E_t = f_{t-h}(Z_t)$, where $f_{t-h}$ is a function determined at time $t-h$ and $Z_t$ an $\mathcal{F}_t$-measurable observation, then a different approach would be to redefine $\tilde{e}_t = e_t \min_z(f_{t+1}(z), \dots, f_{t+h-1}(z))$. This is the strategy implicitly used by \citet{HenziZiegel2021}, and it guarantees that $(\tilde{e}_t)_{t \in \mathbb{N}}$ is an e-process. However, the correction may be very conservative unless the functions $(f_t)_{t \in \mathbb{N}}$ are guaranteed to be bounded away from zero, which in turn limits the power of the test since $\mathbb{E}(f_{t-h}(Z_t) \mid Z_t) \leq 1$ must still hold.

In the following sections we construct sequential e-values for the continuous and for the discrete uniform distribution and for testing stochastic dominance relations with respect to the uniform distribution. General construction principles for e-values, and possible caveats, are explained in the special case of the continuous uniform distribution in Section \ref{subsec:continuous}, but also apply to the other situations in Sections \ref{subsec:discrete} and \ref{subsec:approx_cuf}. An \textsf{R} package implementing the methods is available on GitHub (\url{https://github.com/AlexanderHenzi/epit}), and technical details are given in Appendix \ref{app:implementation}.

\subsection{Continuous uniform distribution} \label{subsec:continuous}
Let $Z$ be a random variable with values in $[0,1]$. We are interested in testing whether $Z$ is uniformly distributed, that is, constructing e-values for the hypothesis
\begin{equation} \label{eq:cuf}
	\mathcal{H}_{\mathrm{CUF}} := \{\mathrm{UNIF}(0,1)\}.
\end{equation}
The underlying set of distributions, $\mathcal{P}$, simply consists of all distributions on the interval $[0,1]$. As a first strategy, we suggest to test $\mathcal{H}_{\mathrm{CUF}}$ against the family of beta distributions which we denote by $\mathcal{H}'$. Any $P \in \mathcal{H}'$ can be parametrized by a vector in the set
\[
	\Theta=\{(\alpha, \beta) \in \R^2 \mid \alpha > 0, \hspace{0.1cm}\beta >0 \}.
\] 
Let $P_{(\alpha, \beta)}$ denote the beta distribution with parameters $(\alpha, \beta)$, so that $\mathcal{H}_{\mathrm{CUF}}=\{P_{(1,1)}\}$. As mentioned in the introduction, the hypothesis $\mathcal{H}_{\mathrm{CUF}}$ is simple, and for any $(\alpha, \beta) \neq (1,1)$ the density, or likelihood ratio, with respect to $\mathrm{UNIF}(0,1)$,
\[
	E^{\alpha,\beta}(Z) := \frac{1}{B(\alpha,\beta)} Z^{\alpha - 1}(1-Z)^{\beta - 1}
\]
is an e-value testing $\mathcal{H}_{\mathrm{CUF}}$ against $\{P_{(\alpha, \beta)}\}$, where $B(\cdot,\cdot)$ denotes the beta function. \citet{GrunwaldHeideETAL2019} suggest to determine e-values in such a way that the expected logarithm of the e-value is maximal in the worst case scenario, and refer to e-values with this property as growth rate optimal in worst case (GROW). Following this criterion, parameters $(\alpha^*, \beta^*) \in \Theta$ would have to be found such that
\begin{equation*}
	\inf_{(\alpha, \beta)\in \Theta}\E_{Z \sim P_{(\alpha, \beta)}}[\log(E^{\alpha^*, \beta^*}(Z))]
\end{equation*}
is maximal. However, this approach is only feasible if either $\alpha$ or $\beta$ (or their ratio or difference) is fixed, which yields a one-parameter exponential family for which results of \citet{GrunwaldHeideETAL2019} are applicable. In many situations this is a prohibitive limitation, since no sufficient prior knowledge is available to restrict the parameters. On the other hand, if both $\alpha$ and $\beta$ can take any positive values, the GROW e-value is constant $1$, because the infimum in the equation above is negative unless $\alpha^* = \beta^* = 1$.

As a different strategy in sequential settings, we propose to estimate $(\alpha,\beta)$ by maximum likelihood estimation (MLE) to optimize power for the next e-value, in the spirit of the betting strategies suggested by \citet{Waudby-SmithRamdas2020} for estimating a bounded mean. Given a sequence of observations $(z_t)_{t \in \N} \subseteq[0,1]$, one can successively calculate e-values for $\mathcal{H}_{\mathrm{CUF}}$ as follows: For $t\geq2$ estimate parameters $(\hat{\alpha}_{t},\hat{\beta}_{t})$ by MLE, that is
\begin{equation} \label{eq:mle}
	(\hat{\alpha}_{t},\hat{\beta}_{t})=  \argmax_{(\alpha, \beta)\in \Theta}\sum_{i=1}^{t} \log\big(p_{(\alpha, \beta)}(z_i)\big),
\end{equation}
where $p_{(\alpha, \beta)}$ denotes the Lebesgue density of $P_{(\alpha, \beta)}$. Set $E_1= E_{2}=1$ and calculate 
\begin{equation} \label{eq:beta_e}
	E_{t+1}= E^{\hat{\alpha}_t,\hat{\beta}_t}(z_{t+1})
\end{equation}
to obtain a sequence $(E_t)_{t \in \N}$ of e-values for testing the null hypothesis that the sequence $(z_t)_{t \in \N}$ is i.i.d.~$\mathrm{UNIF}(0,1)$.

To construct e-values at lag $h$, parameter estimation can be performed separately on all all subsamples with indices $\{k + hs \mid s = 0, 1, \dots\}$, $k=1, \dots, h$. That is, for $t\geq 2h$ calculate 
\begin{equation} \label{eq:mle_lag}
	(\hat{\alpha}^k_{t},\hat{\beta}^k_{t})= \argmax_{(\alpha, \beta)\in \Theta}\sum_{s: k + hs \leq t} \log\big(p_{(\alpha, \beta)}(z_{k+hs})\big), \quad k=1, \dots, h.
\end{equation}
Set $E_1=\dots = E_{2h}=1$ and, for $t = h, h + 1, \dots$, 
\[
	E_{k + th}= E^{\hat{\alpha}^k_t,\hat{\beta}^k_t}(z_{k + th}), \quad k = 1, \dots, h.
\]
Then $(E_t)_{t \in \N}$ are sequential e-values at lag $h$ for the null hypothesis that $z_t \sim \mathrm{UNIF}(0,1)$ conditional on $z_1, \dots, z_{t-h}$ for all $t$, and these e-values can be combined with the formula \eqref{eq:u_statistics}.

\begin{rem}
Estimating the parameters by maximum likelihood for testing the $\mathrm{UNIF}(0,1)$ hypothesis can be regarded as an instance of the running MLE testing method proposed by \citet[Section 7]{Wasserman2020}, and has in fact already been proposed early on by \citet[Equation 10:10]{Wald1947}. Unlike the GROW criterion, sequentially computing the MLE is a sequential betting strategy which tries to adapt to the current alternative in the available data. The strategy to maximize the expected logarithm of a product is sometimes referred to as Kelly betting in reference to \citet{Kelly1956}, who, however, analyzed simple alternative hypotheses rather than adaptive strategies.
\end{rem}

The beta family of distributions is flexible enough to adapt the most common violations of uniformity which occur in practice, namely increasing, decreasing, unimodal and U-shaped densities. This also covers the typical shapes of the PIT distribution for biased and over- or underdispersed probabilistic forecasts. However, in certain applications or data-rich situations, it may be desirable not to restrict the shape of the e-values to a parametric family. A powerful tool for such cases is kernel density estimation, which allows with a sample $\boldsymbol{\zeta}^k = (\zeta_1, \dots, \zeta_k) \in [0,1]^k$ to approximate any density on the unit interval by a mixture
\[
	E^{K,b, \boldsymbol{\zeta}^k}(Z) = \sum_{i = 1}^k \frac{1}{b}K\left(\frac{Z - \zeta_i}{b}\right),
\]
where $K$ is a suitable kernel density and $b > 0$ the bandwidth. The selection of the bandwidth, and even of the kernel $K$, can be done in a sequential fashion like the parameter estimation for the beta e-values. For the e-value at time $t$, the sample $\boldsymbol{\zeta}^{t-1}$ can be taken as $(z_1, \dots, z_{t-1})$ for lag $1$ forecasts, and the e-value $E^{K,b, \boldsymbol{\zeta}^{t-1}}$ is evaluated at the observation $z_t$. For higher lags, the procedure is separated by subsamples with lag $h$ like for MLE in \eqref{eq:mle_lag}.

Compared to the e-values based on beta distributions, the kernel density approach offers more flexibility, which on the other hand also implies more implementation decisions, especially due to the complicating fact that the domain $[0,1]$ is a bounded interval, which calls for boundary corrected density estimates. We describe our implementation in Appendix \ref{app:implementation}. Furthermore, while MLE for parameter estimation in the beta e-values is theoretically motivated by maximizing the growth rate in sequential settings, estimation methods for kernel densities are often based on different criteria, such as integrated mean squared error, which do not have a natural interpretation in the context of e-values. Nevertheless, our simulations illustrate that this approach may still yield good results in practice. 


In the practical implementation of e-values, some details should be taken into account. Under the null hypothesis \eqref{eq:cuf}, the boundary points $0$ and $1$ occur with probability zero, but in applications, observations of exactly $0$ or $1$ appear in most datasets, for example due to rounding. This may be problematic for the construction of the e-values (for example, the estimator \eqref{eq:mle} diverges if $z_i \in \{0, 1\}$ for some $i$), and it may lead to e-values equal to zero or infinity. In our implementation, we decided to ignore observations in $\{0, 1\}$ both in the parameter estimation and when computing the e-values; the latter corresponds to setting $E^{\alpha, \beta}(z) = E^{K, b, \boldsymbol{\zeta}}(z) = 1$, for $z \in \{0, 1\}$, which is a valid strategy since it does not change the expectation of the e-value under the null hypothesis. The rationale is that if zeros or ones occur only rarely, then omitting them should not influence the results. On the other hand, if they occur frequently then it is questionable whether a test of the $\mathrm{UNIF}(0,1)$ hypothesis is really necessary in the given problem since the null hypothesis is obviously false.

A second practical issue is that e-values of exactly zero should be prevented since the e-values lose their power once a level of zero is reached. In the betting interpretation of \citet{Shafer2021}, this would mean that all capital for betting against the null hypothesis is lost. For the beta distributions, zeros can only occur when $Z \in \{0, 1\}$, but the kernel e-values may be zero also inside $(0,1)$ when there are no data points in some region. A simple correction is to replace the e-values $(E_t)_{t \in \mathbb{N}}$ by convex combinations $(\lambda_t + (1-\lambda_t) E_t)_{t \in \mathbb{N}}$ for some $\lambda_t > 0$. We set $\lambda_t = 1/t$ in our implementation, since the danger of zero e-values is typically larger for smaller sample sizes, where the sequential parameter estimation is less stable or observations may be sparse in some subsets of $(0,1)$. When constructing e-values sequentially, one may also set the first $n_0$ e-values to $1$ and start the sequential parameter estimations with a slightly larger sample size, which increases stability. We set $n_0 = 10$ for both the beta and kernel e-values; the minimum $n_0$ to perform MLE for the beta e-values is $n_0 = 2$.

\subsection{Discrete uniform distribution} \label{subsec:discrete}
For $m\geq 1$ the null hypothesis in the discrete case is
\[
	\mathcal{H}_{\mathrm{DUF}}:= \big\{\mathrm{UNIF}(\{1, \dots, m\})\big\},
\]
and the underlying set $\mathcal{P}$ consists of all probability distributions on $\{1, \dots, m\}$. Any $P \in \mathcal{P}$ can be parametrized by $m$ weights in the set $\{ \boldsymbol{w} \in [0,1]^{m} \mid \sum_{i=1}^m w_i = 1 \}$, and $\mathcal{H}_{\textrm{DUF}}=\{ P_{\boldsymbol{w}_0}\}$ for $ \boldsymbol{w}_0=(1/m)_{i=1}^{m}$. Let $R$ be a random variable with values in $\{1, \dots, m\}$. Since $\mathcal{H}_{\textrm{DUF}}$ is a simple null hypothesis, the likelihood ratio
\[
	E(R)=\frac{p_{1}(R)}{p_0(R)}=m\hspace{1mm} p_{1}(R)
\]
is an e-value testing $\mathcal{H}_{\textrm{DUF}}$ against the simple alternative hypothesis $\{P_1\}$, where $P_1 \in \mathcal{P}$ has probability mass function $p_1$. Like in the continuous case, we suggest a parametric and a nonparametric method for constructing $p_1$ sequentially.

For parametric e-values, we propose to use the beta-binomial probability mass function
\[
	p^{\alpha, \beta}(r) = \binom{m - 1}{r - 1}\frac{B(\alpha - r + 1, \beta + m - r)}{B(\alpha, \beta)}, \quad \alpha, \beta > 0
\]
with support in $\{1, \dots, m\}$. This yields e-values with properties similar to the beta e-values, and estimation can again be performed sequentially with the maximum likelihood method. Like the beta distribution on $[0,1]$, the beta-binomial distribution can approximate increasing, decreasing, unimodal and U-shaped probability mass functions on $\{1, \dots, m\}$.

The most natural nonparametric method for obtaining $p_1$ is the empirical distribution. That is, for a given sample $r_1, \dots, r_{t} \in \{1, \dots, m\}$, $p_1(R) = p_{1;t}(R)$ can be set as the empirical frequency of $R$ in the sample up to time $t$, and at time $t + 1$, the frequencies are updated accordingly with the value of $r_{t+1}$. A drawback of this procedure is that the e-values may attain zero if one of the frequencies $p_{1;t}(j)$, $j = 1, \dots, m$, is zero. To prevent this, one may start with a particular $\mathcal{P}_{1} \in \mathcal{P}$, which serves as a first guess for what the actual frequencies will look like. For example, a neutral first guess is $P_1 = P_{\boldsymbol{w}_0}$, and at time $t$, the weights could be updated with the formula
\[
	\boldsymbol{w}_t = \left(\frac{k_1^t+1}{t +m}, \dots, \frac{k_m^t+1}{t + m} \right),
\]
where $k_j^t = \#\{i = 1, \dots, t \mid r_i = j\}$. Here we successively update with the empirical distribution and each component of the weight vector contains one artificial observation. In comparison with the beta-binomial weights, it has to be expected that for even moderate $m$ (say, $20$ or $50$, as common in ensemble forecasting), much larger sample sizes are required to recover the actual underlying distribution.

\subsection{Stochastic ordering with respect to the uniform distribution}\label{subsec:approx_cuf}
Instead of testing whether $Z \in [0,1]$ is distributed according to $\textrm{UNIF}(0,1)$, one is sometimes only interested in whether it attains systematically lower or higher values than expected under $\textrm{UNIF}(0,1)$. This is formalized by the hypotheses
\begin{align}
	\mathcal{H}_{\mathrm{ST}} & = \{P\in\mathcal{P}([0,1] ) \;|\; P \leST \mathrm{UNIF}(0,1)\}, \label{eq:stoch_smaller} \\
	\overline{\mathcal{H}}_{\mathrm{ST}} & = \{P\in\mathcal{P}([0,1] ) \;|\; P \geST \mathrm{UNIF}(0,1)\}, \label{eq:stoch_greater}
\end{align}
where $\mathcal{P}([0,1])$ is the set of all distributions on $[0,1]$. The quantile forecasts described in Section \ref{sec:preliminaries} give one motivation to test these hypotheses. More generally, for a random variable $Y$ and a strictly increasing CDF $G$, tests for $\mathcal{H}_{\textrm{ST}}$ or $\overline{\mathcal{H}}_{\mathrm{ST}}$ applied to $G(Y)$ allow to evaluate if the distribution of $Y$ is stochastically smaller or greater than $G$. Note that testing whether a random variable $Z \in [0,1]$ has distribution in $\mathcal{H}_{\mathrm{ST}}$ is equivalent to testing whether the distribution of $1-Z$ lies in $\overline{\mathcal{H}}_{\mathrm{ST}}$.

The null hypotheses $\mathcal{H}_{\mathrm{ST}}$ and $\overline{\mathcal{H}}_{\mathrm{ST}}$ are composite hypotheses, so the construction of e-values is more involved than for the continuous and discrete uniform distribution. The following result characterizes e-values for $\mathcal{H}_{\textrm{ST}}$ and $\overline{\mathcal{H}}_{\mathrm{ST}}$ under the additional restriction that they have expected value $1$ under the uniform distribution, i.e.~are Lebesgue densities.

\begin{prop} \label{prop:stochastic_test}
Let $f$ be a Lebesgue density on $[0,1]$. Then $\mathbb{E}_{P}(f(Z)) \leq 1$ for all $P \in \mathcal{H}_{\mathrm{ST}}$ ($P \in \overline{\mathcal{H}}_{\mathrm{ST}}$) if and only if there exists an increasing (decreasing) density $\tilde{f}$ and a Lebesgue null set $A$ such that $f(x) = \tilde{f}(x)$ for all $x \not\in A$ and $f(x) < \tilde{f}(x)$ for all $x \in A$.
\end{prop}

\begin{rem}
\citet[Section 2]{Vovk2021} call random variables $p \in [0,1]$ which satisfy $P(p \leq \alpha) \leq \alpha$ for all $\alpha \in (0,1)$ \emph{p-variables}, and a decreasing function $f: [0,1] \mapsto [0, \infty]$ a \emph{p-to-e calibrator} if $f(p)$ is an e-value for all p-variables $p$. In simple words, a p-to-e calibrator is a function which transforms p-values into e-values. This is closely related to the stochastic dominance hypotheses in this section. The set $\overline{\mathcal{H}}_{\mathrm{ST}}$ contains the distributions of all p-variables, and Proposition \ref{prop:stochastic_test} states that decreasing functions are indeed the only p-to-e-calibrators, when additionally expected value $1$ under the $\mathrm{UNIF}(0,1)$ distribution is imposed. For testing $\mathcal{H}_{\mathrm{ST}}$ ($\overline{\mathcal{H}}_{\mathrm{ST}}$), one should always take $f$ right continuous (left continuous) and the set $A$ empty, as otherwise $f(Z)$ is not an admissible e-value in the sense of \citet{Ramdas2020}, or $f$ not an admissible calibrator according to the definition of \citet{Vovk2021}. Notice that \citet{Vovk2021} would additionally require $f(1) = \infty$ for admissibility (or $f(0) = \infty$ in the case of $\overline{\mathcal{H}}_{\mathrm{ST}}$), but in practice one might want not to immediately reject the null hypothesis even though a single observation of $1$ is impossible under $\mathcal{H}_{\mathrm{ST}}$.
\end{rem}

For constructing e-values in a sequential setting, a suitable estimator for decreasing (or increasing) density functions is the Grenander estimator \citep{Grenander1956}, which is the maximum likelihood estimator among all decreasing density functions. The Grenander estimator produces piecewise constant density functions, and as a smooth alternative, we propose the estimator by \citet{Turnbull2014} based on mixtures of Bernstein polynomials, that is, beta densities. This estimator was originally proposed for the estimation of unimodal densities, but monotone densities can be easily accommodated by setting the mode to zero or one. Estimation is based on minimizing a squared distance between the ECDF of a sample $z_1, \dots, z_n$ under constraints on the mixture weights to ensure monotonicity. Different from the Grenander estimator, there is a tuning parameter, namely the maximum degree in the Bernstein polynomials, for which \citet{Turnbull2014} propose several selection criteria. Sequential updating of the estimator, the construction of lag $h$ e-values, and potential corrections to avoid e-values of zero can be done as described for the case of the $\mathcal{H}_{\mathrm{CUF}}$ hypothesis in Section \ref{subsec:continuous}.

The Grenander estimator has the additional advantage that it automatically adapts in the case when it is known (or cannot be excluded) that the distributions of interest have discrete support. To see this for the hypothesis $\mathcal{H}_{\mathrm{ST}}$, assume that the support is $0 = s_1 < \dots < s_k < s_{k + 1} = 1$; here $s_1 = 0$ and $s_k < 1$ are necessary conditions for $P \in \mathcal{H}_{\mathrm{ST}}$. If $f$ is an increasing density, then $\mathbb{E}_{P}(f(Z)) \leq 1$ by Proposition \ref{prop:stochastic_test}, but the piecewise constant density $g = g(z; f)$ defined by
\begin{equation} \label{eq:constant_approx_density}
	g(z; f) = \frac{\int_{s_i}^{s_{i + 1}} f(z) \diff z}{s_{i + 1} - s_i}, \, z \in [s_i, s_{i+ 1}), \, i = 1, \dots, k - 1, \ g(z; f) = \frac{\int_{s_i}^{s_{i + 1}} f(z) \diff z}{1 - s_k}, \, z \geq s_k,
\end{equation}
is also increasing and satisfies $g(s_i; f) \geq f(s_i)$, $i = 1, \dots, k$, so $g(\cdot; f)$ yields a more powerful e-value than $f$. If $f$ is computed with the Grenander estimator and all observations are in $\{s_1, \dots, s_k\}$, then $f$ is already piecewise constant on the intervals $[s_i, s_{i + 1})$, and therefore $g(z; f) = f(z)$. The density \eqref{eq:constant_approx_density} can also be interpreted as the likelihood ratio between the probabilities $g_i = \int_{s_i}^{s_{i + 1}} f(z) \diff z$ and the discretization of the uniform distribution which puts mass $p_i = s_{i + 1} - s_i$ on the points $s_i$.

\begin{rem}
Assume that we are interested in the hypothesis 
\begin{equation}
	\mathcal{H}= \big\{\boldsymbol{P} \in \mathcal{P}\big([0,1]\times [0,1]\big) \mid P_1 \leST \textrm{UNIF}(0,1) \leST P_2\big\},
\end{equation}
where $\mathcal{P}=\mathcal{P}([0,1]\times [0,1])$ denotes the set of all bivariate distributions on $[0,1]\times [0,1]$ and $P_1, P_2$ denote the marginal distributions of some $\boldsymbol{P} \in \mathcal{P}$. Then
\[
	\mathcal{H} = \big\{\boldsymbol{P} \in \mathcal{P} \mid P_1 \leST \mathrm{UNIF}(0,1) \big\} \cap \big\{\boldsymbol{P} \in \mathcal{P} \mid  \mathrm{UNIF}(0,1) \leST P_2\big\}
		= \mathcal{H}_{\mathrm{ST};1} \cap  \overline{\mathcal{H}}_{\mathrm{ST};2}.
\]
Since we can write $\mathcal{H}$ as an intersection of two hypotheses it follows immediately that $(E_1 + E_2)/2$ is an e-value for $\mathcal{H}$ if $E_1, E_2$ are e-values for $\mathcal{H}_{\mathrm{ST};1},\overline{\mathcal{H}}_{\mathrm{ST};2}$ respectively. E-values for $\mathcal{H}_{\mathrm{ST};1}, \overline{\mathcal{H}}_{\mathrm{ST};2}$ can be constructed with the methods proposed in this section, since the hypotheses only impose restrictions on one of the marginals.
\end{rem}
\begin{exmp} \label{exmp:e_quantile_pit}
In this example we show how to use the e-values of Proposition \ref{prop:stochastic_test} and the above remark to check probabilistic calibration of quantile forecasts as defined in Section \ref{sec:preliminaries}. Assume that for given quantile levels $0 < \alpha_1 < \dots < \alpha_K < 1$ we sequentially predict quantiles $(q_{1;t}, \dots, q_{K;t})_{t \in \mathbb{N}}$ at lag 1 and observe the quantities $(y_t)_{t \in \mathbb{N}}$. We calculate the sequence of upper quantile PITs $(z_t)_{t \in \mathbb{N}}$ and lower quantile PITs $ (\overline{z}_t)_{t \in \mathbb{N}} \subseteq [0,1]$, where
\[
		z_t= Z_{F_{u;t}}(y_{t+1}) \quad \textrm{and}\quad \overline{z}_t = Z_{F_{\ell;t}}(y_{t+1}).
\]
For $t \geq 1$ and upper quantile PIT values $z_1, \dots, z_t$ we estimate an increasing density $f_t$. Analogously, we estimate a decreasing density $\overline{f}_t$ with the lower quantile PIT $\overline{z}_1, \dots, \overline{z}_t$. By Proposition \ref{prop:stochastic_test}, $E_{t+1}=f_t(z_{t+1})$ is an e-value for $\mathcal{H}_{\textrm{ST;1}}$ and $\overline{E}_{t+1}=\overline{f}_t(z_{t+1})$ is an e-value for $\overline{\mathcal{H}}_{\mathrm{ST};2}$. Sequential e-values for probabilistic calibration of the quantile forecasts are obtained by $\bar{E}_t = (E_t + \overline{E}_t) / 2$, as explained in the above remark. For for $h > 1$, we refer to the usual procedure where we have to estimate densities separately on subsamples with indices $\{k + hs \mid s = 0, 1, \dots\}$ for $k = 1, \dots, h$.
\end{exmp}

\section{Simulation study} \label{sec:simulation_study}
To evaluate the power of the e-values, we generate independent observations $Y \sim \mathcal{N}(0,1)$ and define forecasts $F = \mathcal{N}(\varepsilon, 1 + \delta)$, where $\epsilon, \delta \in \{-0.5, -0.4, \dots, 0.5\}$ are the bias and dispersion error, respectively. Figure \ref{fig:pit_illustration} illustrates the distribution of the PIT $Z_F(Y) = F(Y)$ for different combinations of bias and dispersion error. For $\delta = \varepsilon = 0$ the PIT is uniformly distributed. To obtain comparable simulations for testing the discrete uniform distribution, we generate $20$ independent ensemble forecasts $\boldsymbol{X} = (X_1, \dots, X_m)$ according to $F$, and test for uniformity of $\mathrm{rank}_{\boldsymbol{X}}(Y) \in \{1, \dots, 21\}$. The tests for stochastic order are applied to the PIT $Z_F(Y)$, and we only test if the distribution of $Z_{F}(Y)$ is stochastically greater than $\mathrm{UNIF}(0,1)$. For testing calibration of quantile forecasts, we take $K = 19$ equispaced quantiles (levels $0.05, 0.1, \dots, 0.95$) of the distribution $F$ and compute the e-values as described in Example \ref{exmp:e_quantile_pit}. Since both $F$ and the distribution of $Y$ are absolutely continuous, the lower and upper quantile PITs are discrete in this case with values in $\{0.05, 0.1, \dots, 1\}$ and $\{0, 0.05, \dots, 0.95\}$, respectively.

We display the result of our simulation experiments with test power heat matrices; see Figure \ref{fig:simulations_power} and the additional figures in the Supplementary Material. While this graphical display is self-explanatory, we emphasize that it allows to compare test power across several tests with respect to two directions of alternatives at a single glance. Figure \ref{fig:simulations_power} shows the rejection rates of different tests in the simulation examples at a level of $\alpha = 0.05$ with a sample size of $n = 360$. All e-values apply the stopping criterion $\tau = \min(360, \, \inf\{t \geq 1: e_t \geq 1/\alpha\})$, and we refer to Appendix \ref{app:implementation} for implementation details. The results for different values of $\alpha$ and $n$ are qualitatively similar and presented in the Supplementary Material. For the continuous uniform distribution, we compare the beta e-values and the kernel e-values to the Kolmogorov-Smirnov test (abbreviated {\tt ks.test} in the following).\footnote{The quantile PIT has a discrete distribution in this simulation study, but the {\tt ks.test} as implemented in \textsf{R} is still applicable since it applies an asymptotic distribution for the test statistic which is sufficiently precise for the sample sizes considered here. We refer the reader to the detailed description and references in the \textsf{R} documentation of {\tt ks.test}.} While the {\tt ks.test} has a higher power against biased forecasts, it is less sensitive to dispersion errors than both e-values. The beta e-values generally achieve a higher power than the e-values based on kernel density estimation, but this difference becomes smaller for larger sample sizes; see the Supplementary Material. For the discrete uniform distribution, we take the chisquare test for comparison. The e-values based on the betabinomial distribution are most sensitive to violations of uniformity, whereas constructing e-values with the empirical frequencies of the ranks is not powerful for the given simulation, since the empirical distribution only recovers the shape of the underlying distribution very slowly. For testing the null hypothesis that the PIT is stochastically greater than $\mathrm{UNIF}(0,1)$, we apply a one sided version of the {\tt ks.test}, which turns out to be more powerful than the e-values. Nevertheless, the e-values with Bernstein polynomials achieve a similar power when the forecast is underdispersed. For testing calibration of the quantile forecasts, one-sided {\tt ks.test}s are applied to the upper and lower quantile PIT and corrected with the Bonferroni method, so that probabilistic calibration can be rejected if at least one of the corrected p-values is below $0.05$. The e-values based on the Grenander estimator are more sensitive to forecast dispersion errors than the {\tt ks.test}, but less sensitive to the bias. The Bernstein e-values achieve a lower power, which is due to the fact that they do not automatically adapt to the discreteness of the quantile PIT.

To summarize, in all simulations the e-values are able to achieve similar power as established methods when optional stopping is applied. Without optional stopping, i.e.~when only considering the e-value at the end of the observation period instead of the anytime-valid p-value $p_t = (\max_{i=1, \dots, t}e_i)^{-1}$, the rejection rates are lower than for the classical non-sequential tests. For the discrete uniform distribution, we suggest to use the betabinomial e-values unless the sample size is large or the number of distinct values $m$ is small. In stochastic dominance testing with smooth distributions, it is generally better to apply the Bernstein e-values. The Grenander estimator should be preferred for testing calibration of quantile forecasts when both the underlying forecast distribution and the distribution of the outcome are continuous. 

\begin{figure}
\centering
\includegraphics[width = \textwidth]{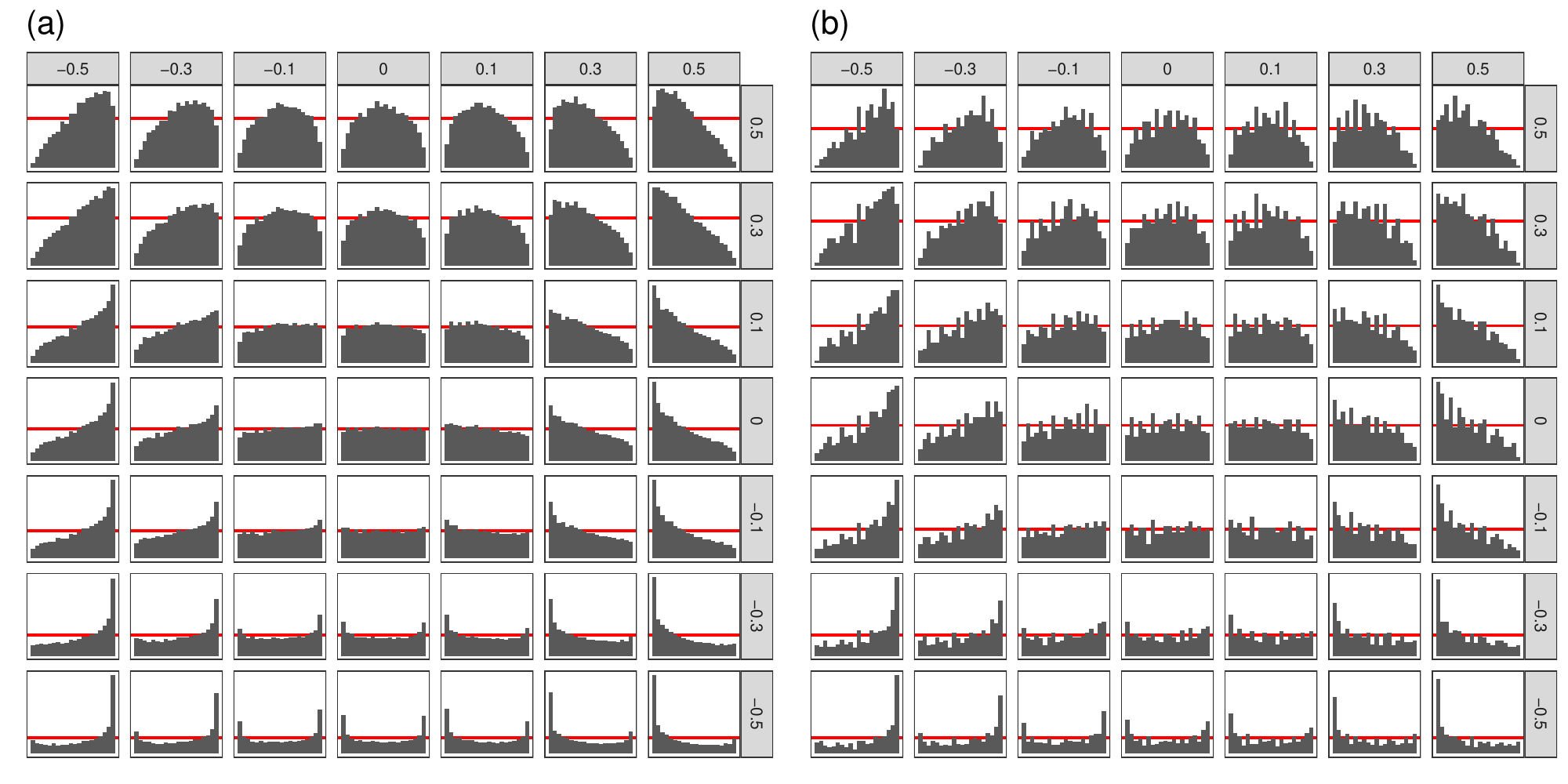}
\caption{Histograms of the PIT in the simulation study, with $20$ equispaced bins and a sample size of (a) $n = 10'000$ (theoretical appearance of the underlying distribution), and (b) $n = 360$ (PIT histogram in a typical simulation). The rows in the figure panels give the dispersion error $\delta$, and the columns give the bias $\varepsilon$. The horizontal line shows the uniform density. Note the different scaling of the y-axis in the panel rows.  \label{fig:pit_illustration}}
\end{figure}

\begin{figure}
\centering
\includegraphics[width = 0.9\textwidth]{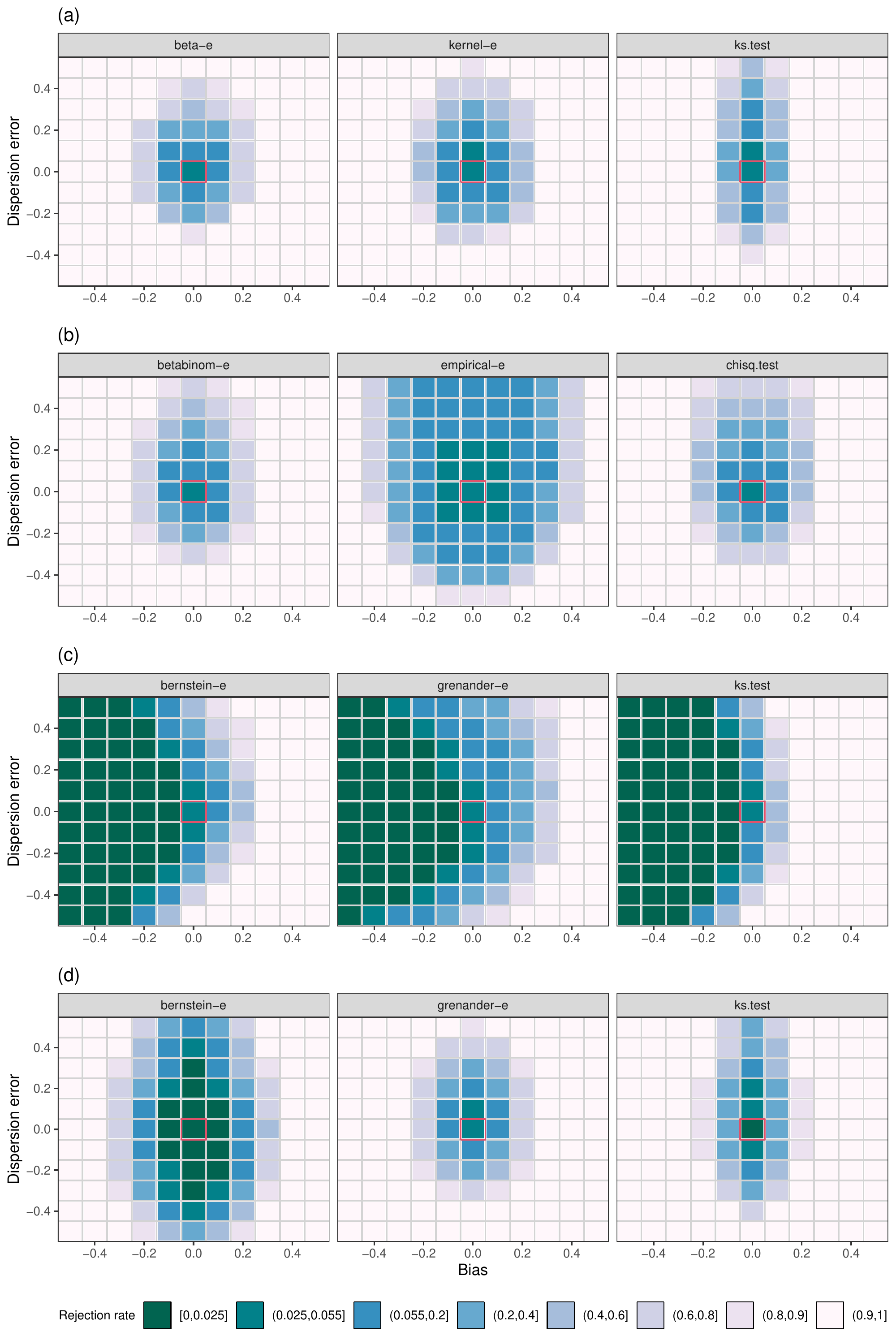}
\caption{Rejection rates of different tests for (a) the continuous uniform distribution (b) the discrete uniform distribution (c) stochastic dominance (d) calibration of quantile forecasts, at the level $\alpha = 0.05$ with a sample size of $n = 360$, depending on the bias and dispersion error. The central box highlights the rejection rates for bias and dispersion error equal to zero. Rejection rates are computed over $5000$ simulations. \label{fig:simulations_power}}
\end{figure}

\section{Case study} \label{sec:case_study}

\subsection{Data and methods}
\begin{table}
\caption{Meteorological station information (latitude, longitude, World Meteorological Organization (WMO) station identifier, station name). \label{tab:station_info}}
\centering
\bigskip
\resizebox{\textwidth}{!}{%
\begin{tabular}{rrrlrrrl}
\toprule
Latitude & Longitude & WMO ID & Name & Latitude & Longitude & WMO ID & Name\\
\midrule
54.18 & 7.90 & 10015 & Helgoland & 51.13 & 13.75 & 10488 & Dresden-Klotzsche\\
53.63 & 9.98 & 10147 & Hamburg-Fuhlsbüttel & 50.87 & 7.17 & 10513 & Köln-Bonn\\
53.65 & 11.38 & 10162 & Schwerin & 50.98 & 10.97 & 10554 & Erfurt-Weimar\\
53.05 & 8.80 & 10224 & Bremen & 49.75 & 6.67 & 10609 & Trier-Petrisberg\\
52.47 & 9.68 & 10338 & Hannover & 50.05 & 8.60 & 10637 & Frankfurt/Main\\
52.13 & 11.60 & 10361 & Magdeburg & 49.77 & 9.97 & 10655 & Würzburg\\
52.38 & 13.07 & 10379 & Potsdam & 49.52 & 8.55 & 10729 & Mannheim\\
52.57 & 13.32 & 10382 & Berlin-Tegel & 48.68 & 9.23 & 10738 & Stuttgart-Echterdingen\\
51.30 & 6.77 & 10400 & Düsseldorf & 49.50 & 11.05 & 10763 & Nürnberg\\
51.50 & 9.95 & 10444 & Göttingen & 49.05 & 12.10 & 10776 & Regensburg\\
51.42 & 12.23 & 10469 & Leipzig/Halle & 48.43 & 10.93 & 10852 & Augsburg\\
\bottomrule
\end{tabular}
}
\end{table}

Ensemble prediction systems have tremendously improved the precision of weather forecasts in the past decades \citep{Bauer2015}. However, it is well known that ensemble forecasts remain subject to biases and dispersion errors, which require statistical correction, so called postprocessing, and a variety of methods is available for this task and applied by weather forecasters \citep{Vannitsem2018}. Ensemble postprocessing methods try to estimate the conditional distribution of the variable of interest given the ensemble forecasts. Postprocessed forecasts usually achieve a better calibration than the raw ensemble forecasts, but they may still be miscalibrated if the relationship between forecasts and observations changes over time or if the postprocessing method (say, a parametric model), is not appropriate for the variable at hand. The PIT is one important tool for identifying misspecification of postprocessed forecasts.

In this case study we apply the e-values to test calibration of postprocessed weather forecasts for 22 SYNOP weather stations in Germany. The dataset is part of the data analysed by \citet{Hemri2014} and was kindly provided by Sebastian Lerch. Forecast data are available through the European Centre for Medium-Range Weather Forecasts (ECMWF) Meteorological Archival and  Retrieval  System  (\url{https://www.ecmwf.int/en/forecasts}) and via TIGGE \citep{Bougeault2010, Swinbank2016}. Station observations can be downloaded from NOAA's Integrated Surface Database (\url{https://www.ncdc.noaa.gov/isd}). Station information is given in Table \ref{tab:station_info}. We postprocess the ensemble predictions from the ECMWF, which consists of 50 perturbed forecasts \citep{Molteni1996, Buizza2005}. The variables considered are 2 meter temperature, wind gust speed, and accumulated precipitation, for lead times of 24, 48, and 72 hours. Data is available from January 1, 2002, to March 20, 2014, and all data until and including the year 2008 is used for training the postprocessing models and the remaining part for validation. The validation dataset consists of $1855$ to $1896$ days per station, slightly varying due to different numbers of missing values.

Postprocessing is performed separately for each forecast lag and for seasons, namely, the model parameters are estimated on data from the calendar months April to September and October to March for forecasts within the respective periods. The postprocessing for all variables is based on the Ensemble Model Output Statistics (EMOS) approach with heteroscedastic regression: The conditional distribution of the variable of interest is approximated by a parametric location-scale family, with the location parameter being an affine function of the ensemble mean and the scale parameter being the exponential of an affine transformation of the ensemble standard deviation. For temperature forecasts, the parametric family are Gaussian distributions. Wind gust speed is modelled with the density of a logistic distribution truncated at zero and rescaled so that it integrates to one. Forecasts for accumulated precipitation are based on the censored logistic distribution, where the probability mass on the non-positive numbers gives the probability of zero precipitation. Parameters are estimated by maximum likelihood for the temperature and wind speed forecasts. For precipitation forecasts, parameters are estimated by minimizing the continuous ranked probability score (CRPS) for precipitation, or by maximum likelihood in case the minimization of the CRPS criterion did not converge. The implementation is in \textsf{R} with the {\tt crch} package \citep{Messner2016}.

To evaluate probabilistic calibration we apply the e-values based on kernel density estimation. To make full use of the large sample size, we use the data of the first year in the validation (more precisely, the first 366 days) only for the computation of a reliable first guess of the density of the PIT, and set all e-values for this period to $1$. For lag 2 and lag 3 forecasts, this gives sample sizes of 183 or 122, respectively, for each of the lagged sequences of e-values. Apart from this modification, the implementation is as described in Appendix \ref{app:kernel}. The e-values based on beta distributions are less powerful to detect misspecification than the kernel densities because the shape of the PIT distribution is sometimes more complicated than only unimodal or U-shaped, which often stems from overlaps of different types of misspecifications in the validation period. We also applied the e-values for the discrete uniform distribution on the raw ensembles, which lead to very fast rejection of the null hypothesis and extremely high e-values (see Table 1 in the Supplementary Material). 

\subsection{Results}
Panels (a) and (b) of Figures \ref{fig:station_10015}, \ref{fig:station_10162} and \ref{fig:station_10729} display the PIT histograms and e-values for selected stations, with the common choice of $20$ bins for plotting the histograms. For many stations, the PIT histograms indicate severe deviations from uniformity, and the e-values give decisive evidence against the null hypothesis of probabilistic calibration. For higher lags, where e-values cannot be merged by product, the power is generally lower than for lag 1. If the goal is purely to check whether the violation of calibration is significant, then Figure \ref{fig:distance_e_value} demonstrates that the e-values indeed correlate well with the distance of the PIT from the uniform density. 

\begin{figure}
\centering
\includegraphics[width = \textwidth]{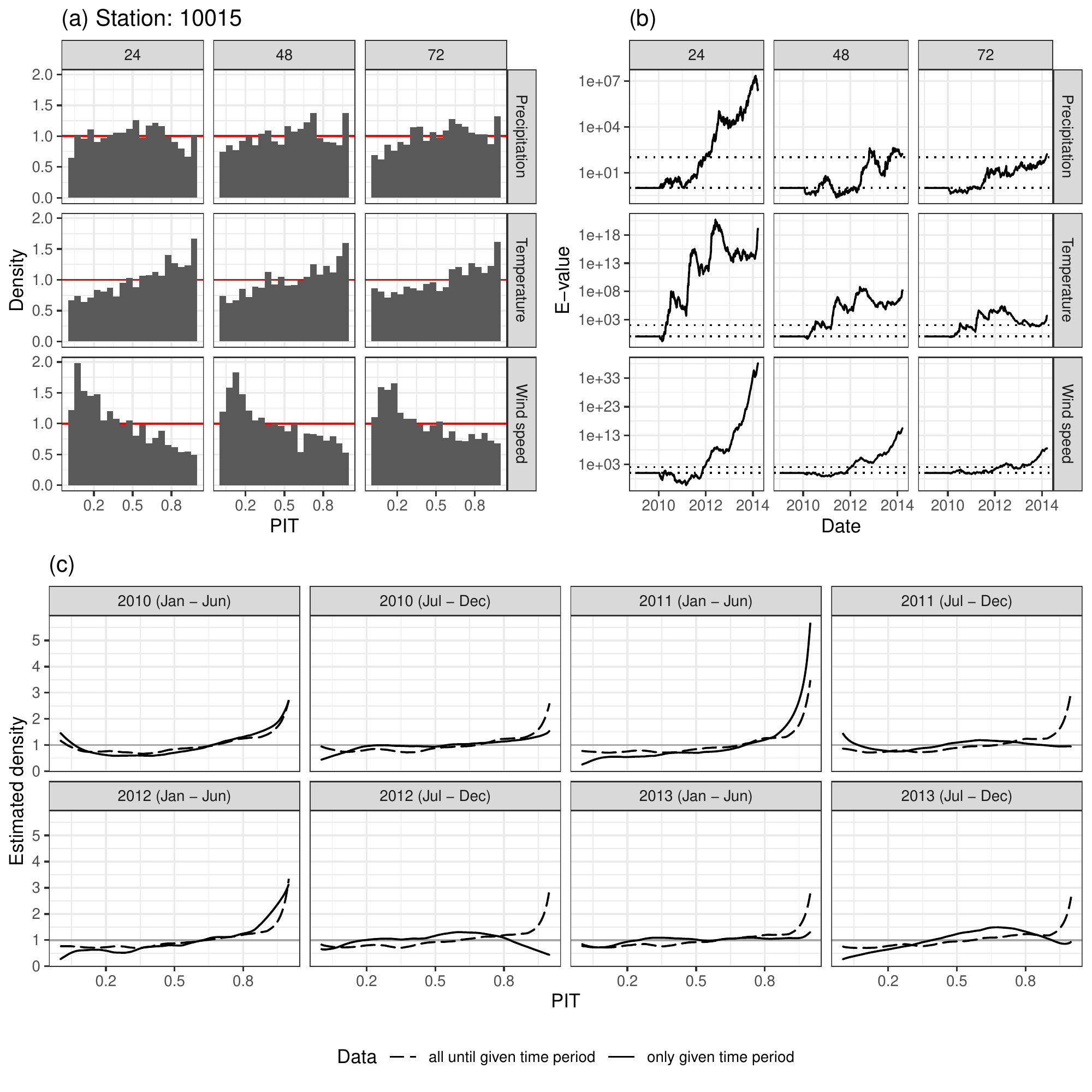}
\caption{(a) PIT histograms of forecasts for station with ID 10015, for all variables and lead times. (b) E-values $(e_t)$ testing uniformity of the PIT of the given forecasts, where the dotted horizontal show the levels $1$ and $100$. (c) Density estimates of the PIT for 24 hour temperature forecasts and the given time periods. The same density estimation method is used as for the computation of the e-values. The dashed density is based on all data until (but not including) the period indicated in the caption, and the solid lines represent the density of the PIT only within the given period. \label{fig:station_10015}}
\end{figure}

As argued in the introduction, evaluating probabilistic calibration only at the end of an observation period is often not informative since forecast misspecification changes over time, and this change of forecast misspecification can indeed be seen in the e-values. Consider first the 24 hour temperature forecasts for station 10015, Helgoland (Figure \ref{fig:station_10015}). The forecasts are biased, with temperatures often being higher than expected under the forecast distribution. Interestingly, the cumulative product of the e-values displayed in panel (b) of Figure \ref{fig:station_10015} exhibits a clear seasonal pattern: Evidence against calibration is usually gained in the first half of each calender year, but not in the second half. To further investigate this effect, we plot the kernel density estimates of the PIT (with the same method as used for constructing the e-values) separated by time periods. Panel (c) of Figure \ref{fig:station_10015} shows for each half year the density of the PIT based on data until (but not including) the given period. For lag 1 forecasts, this is the e-value $E^{K, b, \boldsymbol{\zeta}^t}$, where $\boldsymbol{\zeta}^t$ are all PIT values before the period and $b$ is the bandwidth estimated with data $\boldsymbol{\zeta}^t$. The second density function is estimated based on data within the given time period. For example, the solid line in the second plot in Figure \ref{fig:station_10015} (c) uses PIT values from 2009 until the end of June 2010, and the dashed density is based the PIT from July until December 2010. If the two densities exhibit similar deviations from uniformity, then evidence against the null hypothesis of calibration is gained, since the observed PIT lies in regions where the e-value is greater than one.

It can be seen that the bias of the forecast indeed only occurs in the months January to June, where the e-value increases, but the forecasts are relatively well calibrated from July to December. Improving the postprocessing method should therefore take into account that there is a different seasonal behaviour of the forecasts and observations, which is not captured by performing separate parameter estimation for the months April to September and October to March, and this seasonal behaviour is directly visible in the e-values in panel (b).

\begin{figure}
\centering
\includegraphics[width = \textwidth]{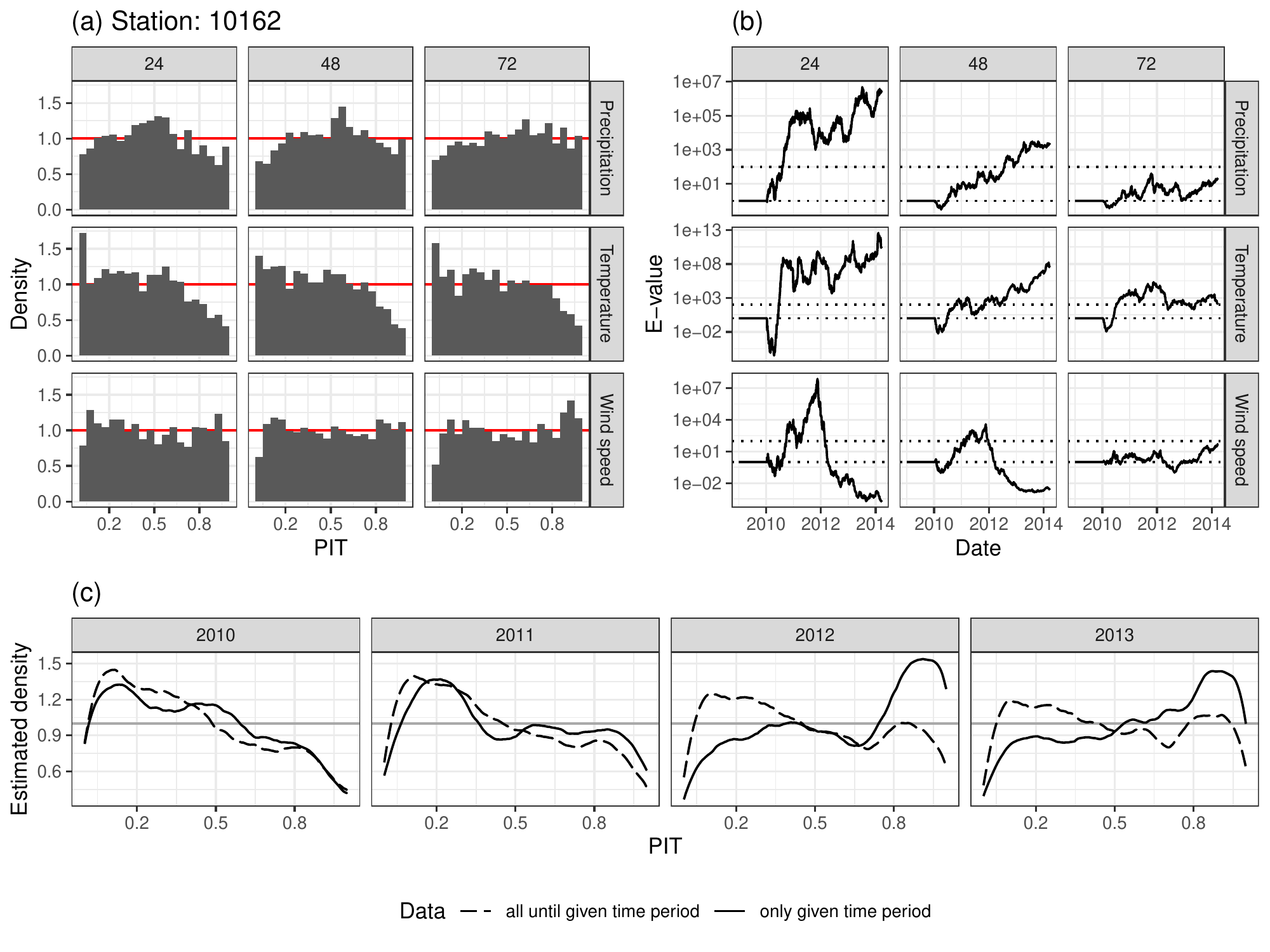}
\caption{Calibration checks for station 10162. The plots are as described in Figure \ref{fig:station_10015}, with panel (c) referring to the 24 hour wind speed forecasts. \label{fig:station_10162}}
\end{figure}

A similar observation can be made for the 24 hour wind speed forecasts for station 10162, Schwerin, in Figure \ref{fig:station_10162}. The PIT histogram looks close to uniform, and the e-value at the end of the observation period is close to zero and therefore suggests that the forecast is calibrated. However, when looking at the full time domain, there is in fact strong evidence against probabilistic calibration: At the end of the year 2011 the e-value reaches a level of more than $10^{7}$, which corresponds to a highly significant p-value of $10^{-7}$. Rejecting calibration based on this observation is statistically valid, because the probability that the process exceeds this level at any time is less or equal to $10^{-7}$. The density estimates, in the same spirit as for the previous station, show that the forecasts are in fact biased over the whole observation period, but the direction of the bias changes at the end of 2011. This change in forecast misspecification is clearly visible in the plot of the e-value over time. The e-value only has power to detect deviations from the null hypothesis which are consistent with the chosen alternative, and hence a decreasing e-value does not necessarily indicate that the forecast is calibrated, but rather that the chosen alternative does not adequately describe the distribution of the PIT. If, like in this application, the density under the alternative is estimated solely on past data, then this typically indicates a change in miscalibration to a misspecification type which has not been observed or was not predominant in the past.

Finally we consider the 48 hour precipitation forecasts for station 10729, Mannheim (Figure \ref{fig:station_10729}). The e-value grows steadily over time and reaches a level of $10^5$, indicating that the underdispersion visible in the PIT is indeed significant. The kernel density estimates in panel (c) of Figure \ref{fig:station_10729} confirm that this underdispersion is consistent over the whole time period and not varying, as one could expect from the plot of the e-values.

To summarize, by examining how e-values develop over time, changes in forecast calibration or miscalibration become visible at a glance. Furthermore, e-values make it is possible to detect forecast miscalibration which cannot be seen directly in a PIT histogram based on the complete data, and yield valid p-values for rejecting calibration at any time point without having to stratify the data in advance. A stratified analysis by season or year, as in the panel (c) of the figures in this section, does of course not necessarily require e-values. However, it has been demonstrated that e-values may simplify this process by indicating whether or at what time points forecast misspecification changes.

\begin{figure}
\centering
\includegraphics[width = \textwidth]{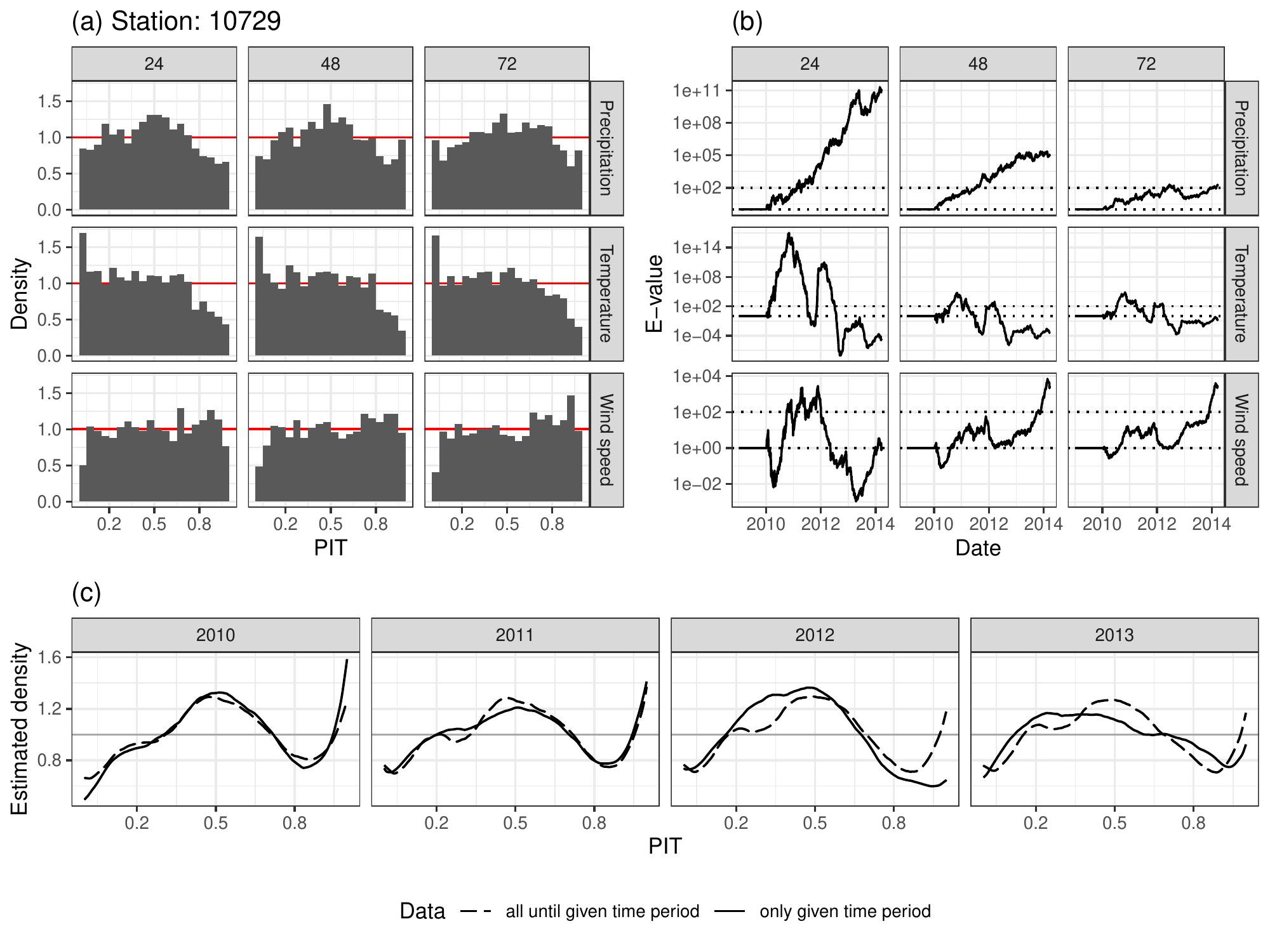}
\caption{Calibration checks for station 10729. The plots are as described in Figure \ref{fig:station_10015}, with panel (c) referring to the 48 hour precipitation forecasts.  \label{fig:station_10729}}
\end{figure}

\begin{figure}
\centering
\includegraphics[width = \textwidth]{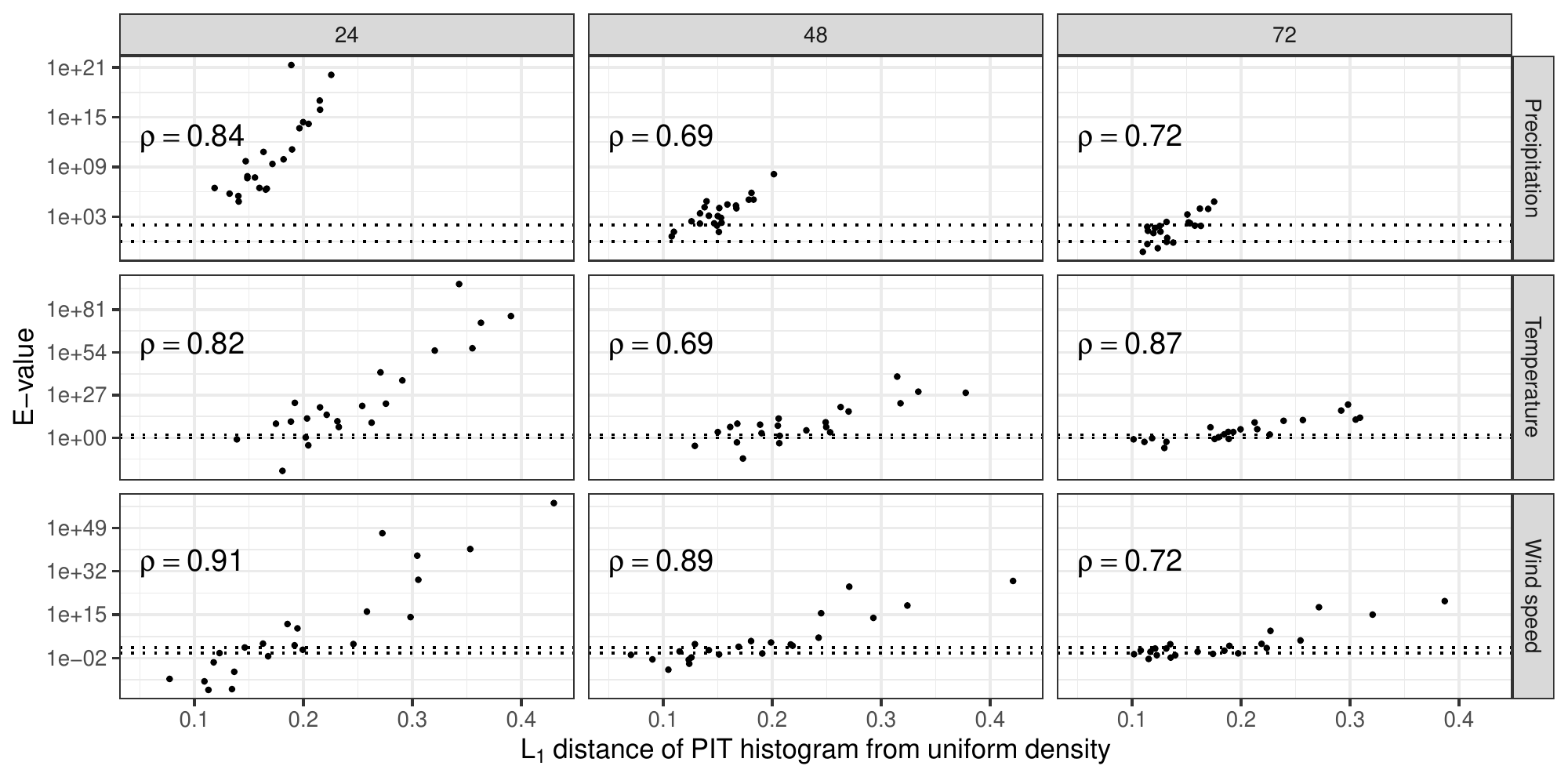}
	\caption{E-values for each station compared to the integrated absolute difference ($\mathrm{L}_1$ distance) between the PIT histogram and the uniform density. The e-values are the ones obtained at the end of the observation period, and $\rho$ gives the Spearman rank correlation between the e-values and the $\mathrm{L}_1$ distance for the given station and lead time. The dotted horizontal lines show the nominal levels of $1$ and $100$. \label{fig:distance_e_value}}
\end{figure}
\section{Discussion}
Forecasting is an inherently sequential task. Most forecasts exhibit non-stationary errors, for example due to seasonal effects, and forecasters adapt and improve their methods and models over time, which results in systematic changes of forecast performance. For this reason forecast evaluation should be sequential as well. Indeed, most practitioners and institutions continuously evaluate the quality of their forecasts; for example, the EMCWF analyses their forecast methods annually in their reports available on \url{https://www.ecmwf.int/en/publications/annual-reports}. From the theoretical side, there is a lack of methods tailored for sequential forecast evaluation, which do not simply rely on a discretization of the time domain and applying static methods for fixed sample sizes.

E-values, which are arguably the suitable tool for sequential forecast evaluation, have received increasing interest in recent years, but so far the methods have not been systematically applied to the evaluation of probabilistic forecasts. We have shown how e-values can be applied to obtain sequentially valid tests for probabilistic calibration, which is one of the most important notions of forecast calibration. The e-values which are provided in this paper are also of stand-alone interest and can be applied in other areas of statistics.

Simulation studies are an important tool to understand rejection rates (power) of newly proposed tests across a range of relevant alternatives. Often, if several parameters are varied in the study, readers are overwhelmed by too many numbers in large tables or too many lines in graphs. We suggest to display summaries of rejection rates as test power heat matrices as given in Figure \ref{fig:simulations_power}. These diagrams allow to see a power comparison of several tests against many alternatives at one glance. 

Our paper focuses on probabilistic calibration. A topic for future work is to derive valid tests for sequential forecast evaluation for other notions of calibration like auto-calibration. In contrast to probabilistic calibration, the notion of auto-calibration extends readily also to multivariate forecasts.

\section*{Acknowledgements}
The authors are grateful to Sebastian Lerch for providing data for the case study and thank Timo Dimitriadis, Tilmann Gneiting and the members of his group for valuable discussions and inputs. Valuable comments by Aaditya Ramdas and an anonymous reviewer helped us to improve this article. This work was supported by the Swiss National Science Foundation. Computations have been performed on UBELIX (\url{https://ubelix.unibe.ch/}), the HPC cluster of the University of Bern.

\bibliographystyle{plainnat}
\bibliography{biblio}

\appendix
\section{Proofs of theoretical results} \label{app:proofs}
\begin{proof}[Proof of Proposition \ref{prop:qpit}]
For $y < q_1$ and $z>q_K$ the conditions $F_u(y) = 0 \le F(y)$ and $F(z)\leq 1= F_\ell(z)$ are always satisfied. If $y \in (q_i,q_{i+1})$ for some $i=1,\dots,K-1$, then
\[
F_u(y) = \alpha_i \le F(q_i) \le F(y) \le F(q_{i+1}-) \le \alpha_{i+1} = F_\ell(y).
\]
For the second claim consider $j \in \{1, \dots, K\}$ and define $\mathcal{I}_j = \{i \in \{1, \dots, K\}\mid q_i = q_j\}$. Then, 
\begin{align*}
& F_u(q_j-)= \min_{i \in \mathcal{I}_j} \, \alpha_{i-1} \leq \max_{i \in \mathcal{I}_j} \, \alpha_i = F_u(q_j), \\
& F_\ell(q_j-)= \min_{i \in \mathcal{I}_j}\,\alpha_{i} \leq \max_{i \in \mathcal{I}_j}\,\alpha_{i+1} = F_\ell(q_j),
\end{align*}
which shows equation \eqref{eq:quantiles}.
\end{proof}

\begin{proof}[Proof of Proposition \ref{prop:stochastic_test}]
We show the claim for $\mathcal{H}_{\textrm{ST}}$. The arguments for $\overline{\mathcal{H}}_{\mathrm{ST}}$ are analogous.

\noindent
Sufficiency: Assume that there exists an increasing function $\tilde{f}$ and a Lebesgue null set $A$ such that $f(x) = \tilde{f}(x)$ for all $x \in [0,1] \setminus A$ and $\tilde{f}(x) > f(x)$ for $x \in A$. Let $P \in \mathcal{H}_{\textrm{ST}}$, then $P \leST \mathrm{UNIF}(0,1)$ and
\begin{equation}
	\E_{P} (f(Z)) \leq \E_{P} (\tilde{f}(Z)) \leq \E_{\textrm{UNIF}(0,1)}(\tilde{f}(Z)) = 1,
\end{equation}
using that $f(x) \leq \tilde{f}(x)$ for all $x \in [0,1]$, isotonicity of $\tilde{f}$, and the fact that $\mathbb{E}_{F_1}(g(X)) \leq \mathbb{E}_{F_2}(g(X))$ for all increasing functions $g$ if $F_1 \leST F_2$.

\noindent
Necessity: Let $f$ be a density on $[0,1]$ such that there exist no Lebesgue null set $A$ and increasing Lebesgue density $\tilde{f}$ such that $f(x) = \tilde{f}(x)$ for $x \not\in A$ and $f(x) < \tilde{f}(x)$ for $x \in A$. We show that then $\mathbb{E}_{P}(f(Z)) > 1$ for some $P \in \mathcal{H}_{\mathrm{ST}}$.

Case 1: There is an increasing Lebesgue density $\tilde{f}$ such that $f(x) = \tilde{f}(x)$ for all $x \not\in A$, where $A$ is a Lebesgue null set. Then there must exist $a \in [0,1]$ such that
\begin{equation} \label{eq:stoch_violated}
	f(a) > \tilde{f}(a).
\end{equation}
If \eqref{eq:stoch_violated} only holds for $a = 1$, then $f(x) \leq \tilde{f}(x)$ for $x \neq 1$ and $f(1) > \tilde{f}(1)$. This yields a contradiction, because with the isotonic function $\check{f}$ defined as $\check{f}(x) = \tilde{f}(x)$ for $x < 1$ and $\check{f}(1) = f(1)$, we have that $f(x) = \check{f}(x)$ for all $x \in ([0,1]\setminus A) \cup \{1\}$, and $f(x) < \check{f}(x)$ for $x \in A \setminus \{1\}$. Hence we can assume that \eqref{eq:stoch_violated} holds for some $a < 1$. If $\tilde{f}(x) \geq f(a)$ for all $x > a$ and for all $a$ such that \eqref{eq:stoch_violated} holds, then similar to before, define $\check{f}(x) = \tilde{f}(x)$ for $x \neq a$ and $\check{f}(a) = f(a)$ for all $a$ for which \eqref{eq:stoch_violated} is true. Then $\check{f}$ is again an increasing function almost surely equal to $f$ and satisfies $\check{f} \geq f$, a contradiction (a figure illustrating this special case can be found in the Supplementary Material). Therefore there must exist $a \in [0,1)$ such that $f(a) > \tilde{f}(b)$ for some $b \in (a, 1]$. This implies $f(a) > \tilde{f}(y)$ for all $y \in [a, b]$, by monotonicity of $\tilde{f}$. Choose $a$, $b$ such that this condition holds, and define the CDF $G$ by
\[
	G(x) = \begin{cases}
		x, & \ x \in [0, a), \\
		b, & \ x \in [a,b), \\
		x, & \ x \in [b, 1].
	\end{cases}
\]
Then $G(x) \geq x$ for $x \in [0,1]$, so $G \in \mathcal{H}_{\mathrm{ST}}$, and
\begin{align*}
	\mathbb{E}_{G}(f(Z)) & = \int_{[0, a)} f(z) \diff z + (b-a) f(a) + \int_{[b, 1]} f(z) \diff z \\
	& = \int_{[0, a)} \tilde{f}(z) \diff z + (b-a) f(a) + \int_{[b, 1]} \tilde{f}(z) \diff z \\
	& > \int_{[0, a)} \tilde{f}(z) \diff z + \int_{[a,b)} \tilde{f}(z) \diff z + \int_{[b, 1]} \tilde{f}(z) \diff z = 1,
\end{align*}
using the fact that $f = \tilde{f}$ Lebesgue almost surely.

Case 2: There exists no monotone increasing Lebesgue density $\tilde{f}$ such that $f(x) = \tilde{f}(x)$ for all $x \in [0,1] \setminus A$, where $A$ is a Lebesgue null set. For $x \in [0,1]$, define $F(x) = \int_0^x f(z) \diff z$. Then $F$ is not convex because otherwise, $F$ would be differentiable almost everywhere and its derivative would be increasing and equal to $f$ for all $x$ not contained in some set of Lebesgue measure zero. This implies that there are points $0 \leq x_1 < x_2 < x_3 \leq 1$ such that
\begin{equation} \label{eq:not_convex}
	\frac{\int_{[x_1, x_2]} f(z) \diff z}{x_2 - x_1} = \frac{F(x_2) - F(x_1)}{x_2 - x_1} > \frac{F(x_3) - F(x_2)}{x_3 - x_2} = \frac{\int_{[x_2, x_3]} f(z) \diff z}{x_3 - x_2}.
\end{equation}
Let $c := (x_3 - x_1) / (x_2 - x_1) = 1 + (x_3-x_2)/(x_2-x_1) > 1$ and define
\[
	G(x) = \begin{cases}
		x, & \ x < x_1, \\
	    x_1 + c(x - x_1), & \ x \in [x_1, x_2), \\
		x_3, & \ x \in [x_2, x_3), \\
		x, & \ x \in [x_3, 1].
	\end{cases}
\]
Then $G(x) \geq x$ for $x \in [0,1]$, and by \eqref{eq:not_convex},
\begin{align*}
	\mathbb{E}_{G}(f(Z)) & = \int_{[0, x_1)} f(z) \diff z + c \int_{[x_1, x_2)} f(z) \diff z + \int_{[x_3, 1]} f(z) \diff z \\
	& = \int_{[0, x_1)} f(z) \diff z + \int_{[x_1, x_2)} f(z) \diff z + \frac{x_3-x_2}{x_2-x_1}\int_{[x_1, x_2)} f(z) \diff z + \int_{[x_3, 1]} f(z) \diff z \\
	& > \int_{[0,1]} f(z) \diff z = 1.
\end{align*}
\end{proof}

\begin{proof}[Proof of Proposition \ref{prop:u_statistics}]

Let $h> 1$, and recall the definitions
\[
	\mathfrak{F} = (\mathcal{F}_t)_{t \in \mathbb{N}}, \quad
		\mathfrak{F}^{[k]} = \left(\mathcal{F}_{\lfloor \frac{t-l}{h}\rfloor h + l}\right)_{t\in \mathbb{N}}, \quad M^{[k]} =  \left(\prod_{l \in I_k(t)} E_l\right)_{t \in \mathbb{N}}, \ k = 1, \dots, h,
\]
with $I_k(t) = \{k + hs : s = 0, \dots, \lfloor (t - k)/h\rfloor\}$. Since $M^{[k]}$ is an $\mathfrak{F}^{[k]}$-supermartingale, it satisfies
\[
	\mathbb{E}_{\mathbb{P}}\left(M^{[k]}_{\tau[k]}\right) \le 1, \ P \in \mathcal{H},
\]
for any $\mathfrak{F}^{[k]}$-stopping time $\tau^{[k]}$. Therefore, for $\mathfrak{F}^{[k]}$-stopping times $\tau^{[k]}$, $k = 1,\dots,h$, 
\[
\mathbb{E}_\mathbb{P}\left(\frac{1}{h}\sum_{k=1}^h M^{[k]}_{\tau^{[k]}}\right) \le 1.
\]
If $\tau$ is an $\mathfrak{F}$-stopping time, then 
\[
\left(\left\lfloor\frac{\tau-k-1}{h}\right\rfloor + 1\right)h + k =: f_k(\tau)
\]
is an $\mathfrak{F}^{[k]}$-stopping time. This implies that for any $\mathfrak{F}$-stopping time $\tau$, we obtain that
\[
\mathbb{E}[M_{\tau + h - 1}] = \mathbb{E}_\mathbb{P}\left(\frac{1}{h}\sum_{k=1}^h M^{[k]}_{f_k(\tau)}\right) \le 1,
\]
because $M_{\tau + h - 1} = \sum_{k=1}^h M^{[k]}_{f_k(\tau)}/h$ for all $t \in \mathbb{N}$.
\end{proof}

\begin{proof}[Proof of the validity of tests based on $\tau_{\alpha,h}$]
For $k = 1, \dots, h$, define
\[
	p^{[k]}_t = \min(1, \, \inf_{s \leq t} 1/M^{[k]}_s), \quad p^{[k]} = \lim_{t \rightarrow \infty} p^{[k]}_t = \min(1, \, 1/\sup_{t \in \mathbb{N}} M^{[k]}_t).
\]
The fact that $M^{[k]}$ is a nonnegative $\mathfrak{F}^{[k]}$-supermartingale implies
\[
	\mathbb{P}(p^{[k]} \leq \alpha) \leq \alpha, \ \alpha \in (0,1], \ P \in \mathcal{H}.
\] 
Consequently, $p^{[k]}$, $k = 1, \dots, h$, are p-variables for $\mathcal{H}$, and can be combined with any of the methods for combining p-variables described in \citet{Vovk_Wang_p_2020}. In particular, the harmonic average of $p^{[1]},\dots, p^{[h]}$ multiplied by $e\log(h)$ is a p-variable, which implies
\[
	P\left(\frac{e\log(h)}{\sum_{k=1}^h1/(hp^{[k]})} \leq \alpha \right) \leq \alpha,
\]
and hence, because $1/p^{[k]} \geq \sup_{t \in \mathbb{N}} M^{[k]}_t$,
\[
	P\left(\frac{1}{he\log(h)}\sum_{k=1}^h \sup_{t \in \mathbb{N}} M^{[k]}_t \geq 1/\alpha \right) \leq \alpha, \ \alpha \in (0,1], \ P \in \mathcal{H}.
\]
\end{proof}

\section{Implementation details} \label{app:implementation}

\subsection{Beta e-values}
The parameters $(\alpha, \beta)$ in the beta e-values are estimated by maximum likelihood with Newton's method for maximization. The moment matching estimator is taken as a starting point, and the Newton iterations are continued until the likelihood between subsequent iterations does not differ by more than $10^{-6}$ or until a maximum number of $20$ iterations is reached. For stability, the values of $(\alpha, \beta)$ are truncated to lie in $[0.001, 100]$, and parameter estimation is only started after $10$ observations are available (the first $10$ e-values are set to $1$). The implementation of Newton's method for maximizing the likelihood uses code adapted from the {\tt Rfast} package \citep{Papadakis2020}.

\subsection{Kernel e-values} \label{app:kernel}
The kernel e-values use the boundary kernel densities as suggested by \citet{Muller1994}. In their original form, these kernel functions may attain negative values, so the non-negativity correction by \citet{Jones1996} is applied. This estimation method is implemented in the {\tt bde} package in \textsf{R} \citep[][function {\tt jonesCorrectionMuller94BoundaryKernel}]{Santafe2015}. The resulting density may sometimes not integrate to one. Therefore, it is evaluated on the discrete grid $0, 0.01, \dots, 0.99, 1$ and rescaled so that this discretized version has integral one. To estimate the bandwidth, the direct plug-in approaches as described in Section 3.6 of \citet{Wand1995} and implemented in the {\tt KernSmooth} package \citep{Wand2021} are applied, with $2$ levels of functional estimation for the plug-in rule. In this article, all results with the kernel e-value are based on the boundary corrected Epanechnikov kernel, and only the bandwidth is updated sequentially.

\subsection{Betabinomial e-values}
For the estimation of the parameters $(\alpha, \beta)$ in the betabinomial e-values, Newton's method is applied to maximize the likelihood. The moment matching estimators are taken as starting point, and iterations are continued until the sum of the absolute differences between the parameter estimates, $|\alpha_k - \alpha_{k - 1}| + |\beta_k - \beta_{k-1}|$, is smaller than $10^{-7}$ or until a maximum number of $20$ iterations is reached. The stopping criterion is different from the estimation of the beta-distribution, since the evaluation of the log-likelihood function for the betabinomial distribution is more costly. The values of $\alpha$ and $\beta$ are truncated to lie in $[0.001, 100]$. Parameter estimation starts with $20$ observations (the first $20$ e-values are set to 1), because the smaller number of $10$ observations, which is applied in the beta e-values, led to diverging parameter estimates in some simulation examples. The implementation of Newton's method for maximizing the likelihood uses code adapted from the {\tt Rfast} package \citep{Papadakis2020}.

\subsection{E-values based on empirical frequencies}
The e-values for the discrete uniform distribution based on the empirical frequencies start with a minimum number of $10$ observations, all previous e-values are set to $1$. For each element of the discrete set, one artificial observation is included at the beginning, so that the frequencies in the $t$-th step equal $(k_j^t + 1) / (m + t)$, $j = 1, \dots, m$, where $k_j^t = \#\{i = 1, \dots, t \mid r_i = j\}$.

\subsection{Grenander e-values}
The e-values based on the Grenander estimator start with a minimum number of $10$ observations. The Grenander estimator is recomputed with each new observation, applying the abridged pool-adjacent violaters algorithm by \citet{Henzi2020}. To avoid e-values of exactly zero, the correction $\tilde{E}_t = 1/t + (1 - 1/t)E_t$ is applied.

\subsection{Bernstein e-values}
The estimation of monotone densities with mixtures of Bernstein polynomials is based on adapted \textsf{R} code by \citet{Turnbull2014}. The mixture weights are computed by minimizing the error defined in Equation (5) in \citet{Turnbull2014}, subject to constraints on the weights to ensure monotonicity. This leads to a quadratic programming problem, which is solved with {\tt osqp} from the identically named \textsf{R} package \citep{Stellato2019}. The {\tt osqp} algorithm is faster and more stable than the quadratic programming solver applied in the original version of the code. The relative and absolute convergence tolerance parameters are set to $10^{-5}$ and the maximum number of iterations to $4000$. A minimum number of $10$ observations is required to compute the e-values, and the first $10$ e-values are set to $1$. The maximal degree of the Bernstein polynomials is fixed at $20$ and not estimated. To avoid zero e-values, the correction $\tilde{E}_t = 1/t + (1 - 1/t)E_t$ is applied.

\beginsupplement
\section{Supplementary material}
\subsection{Simulation example}
The rejection rates in all figures for the simulation study are computed over 5000 simulations.

\medskip\noindent
Figures \ref{fig:simulations_power_n_180_alpha_0.05} and \ref{fig:simulations_power_n_720_alpha_0.05} are like Figure 2 in the article but with sample sizes $n = 180$ and $n = 720$. Figures \ref{fig:simulations_power_n_360_alpha_0.01} and \ref{fig:simulations_power_n_360_alpha_0.1} have sample size $n = 360$ but $\alpha = 0.01$ and $\alpha = 0.1$ instead of $\alpha = 0.05$.

\medskip\noindent
Figure \ref{fig:simulations_power_rhist} shows the rejection rates of the tests for the discrete uniform distribution for ensembles of size $m = 10, \, 20, \, 50$, with $n = 360$ and $\alpha = 0.05$.

\medskip\noindent
Figure \ref{fig:simulations_power_qpit} show the rejection rates of the tests for calibration of quantile forecasts for $K = 9, \, 19$ equispaced quantiles, with $n = 360$ and $\alpha = 0.05$.

\subsection{Case study}
The e-values for testing the discrete uniform distributions have been applied to the rank histograms of the raw ECMWF ensemble forecasts. Both methods (e-values based on the empirical distribution and on the betabinomial distribution) clearly reject uniformity with only a few observations, since the raw ensemble forecasts have strong biases and dispersion errors when evaluated against station observations. Like for the PIT, the first $n_0 = 366$ e-values are set to 1 and the corresponding ranks are used to estimate the rank histogram with the empirical distribution or the betabinomial distribution. Table \ref{tab:rhist} shows how many observations after $n_0 = 366$ are required until the sequential e-values first cross the level $10^{8}$.

\subsection{Illustration for the proof of Proposition \ref{prop:stochastic_test}}
Figure \ref{fig:proof_illustration} illustrates the special case in the proof of Proposition 3.2 where $f$ is almost surely equal to a Lebesgue density $\tilde{f}$ and $f(a) > \tilde{f}(a)$ for some $a \in [0,1)$ but $f(x) \leq \tilde{f}(x)$ for all $x > a$. In this case, the function $\check{f}$ defined by $\check{f}(a) = f(a)$ for all such $a$ and $\check{f}(x) = \tilde{f}(x)$ for all other $x$ is increasing, almost surely equal to $f$, and $\check{f}(x) \geq f(x)$ for all $x \in [0,1]$.

\begin{figure}
	\centering
	\includegraphics[width = 0.9\textwidth]{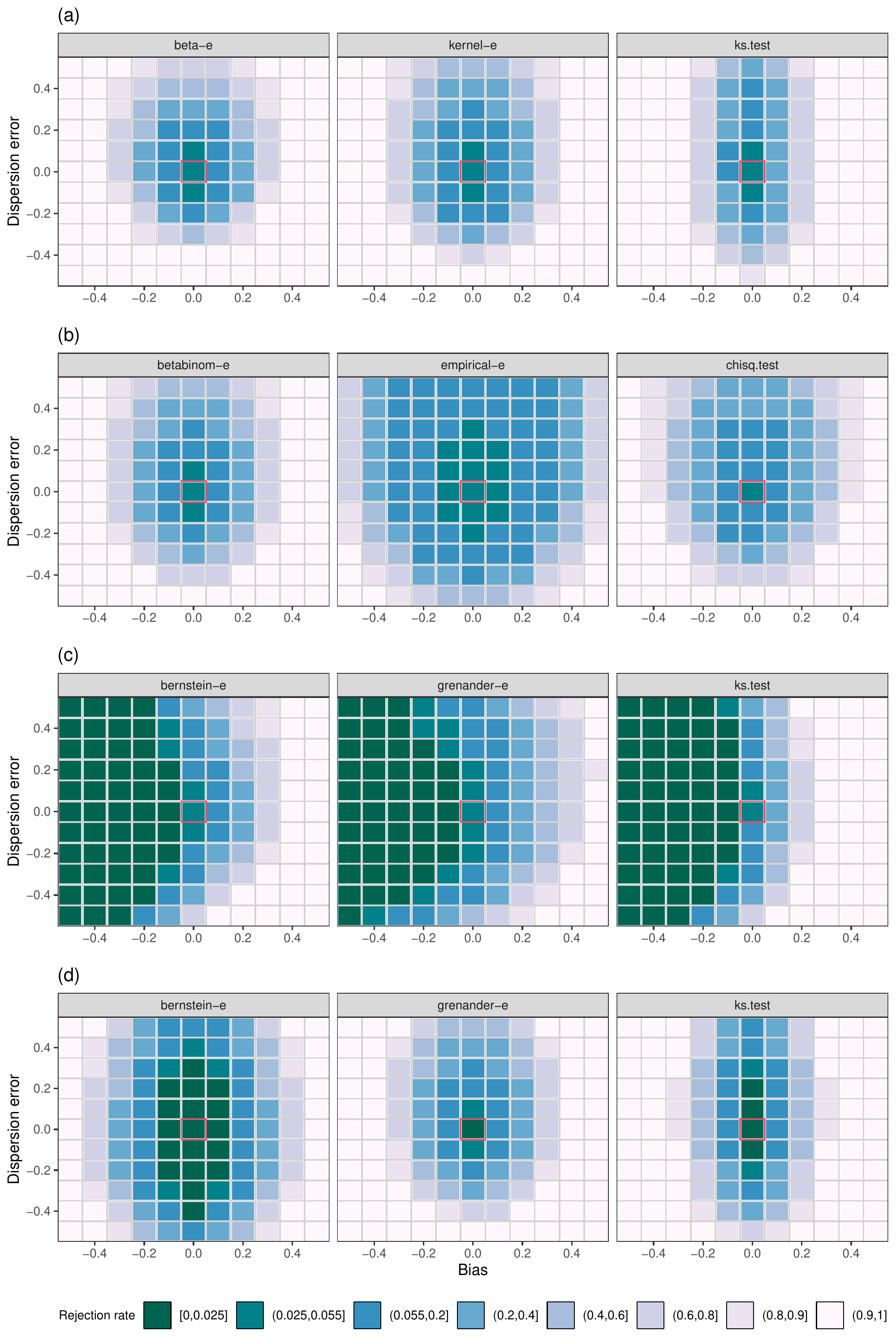}
	\caption{Rejection rates of different tests for (a) the continuous uniform distribution (b) the discrete uniform distribution (c) stochastic dominance (d) calibration of quantile forecasts, at the level $\alpha = 0.05$ with a sample size of $n = 180$, depending on the bias and dispersion error. The red box highlights the rejection rates for bias and dispersion error equal to zero. \label{fig:simulations_power_n_180_alpha_0.05}}
\end{figure}

\begin{figure}
	\centering
	\includegraphics[width = 0.9\textwidth]{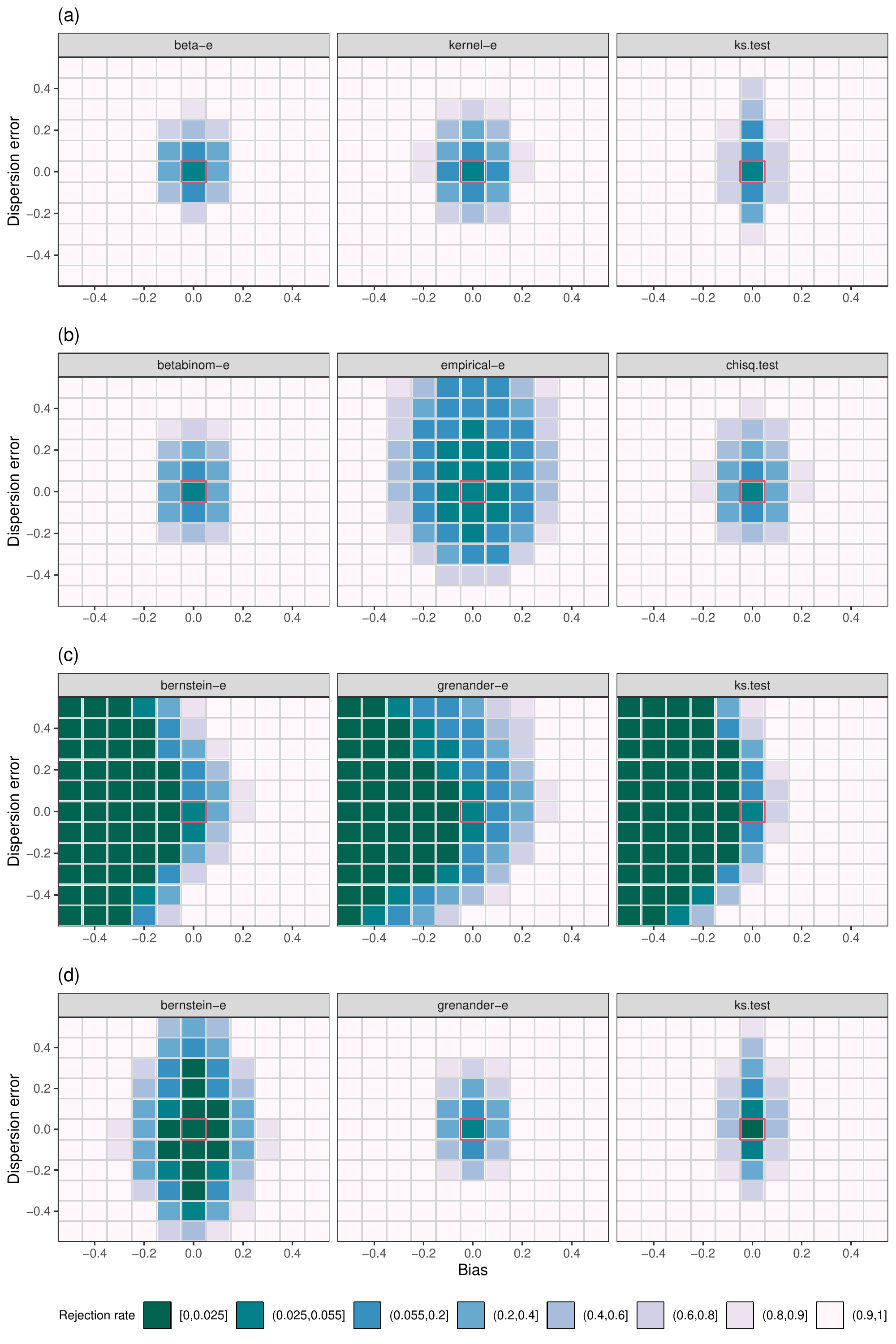}
	\caption{Rejection rates of different tests for (a) the continuous uniform distribution (b) the discrete uniform distribution (c) stochastic dominance (d) calibration of quantile forecasts, at the level $\alpha = 0.05$ with a sample size of $n = 720$, depending on the bias and dispersion error. The red box highlights the rejection rates for bias and dispersion error equal to zero. \label{fig:simulations_power_n_720_alpha_0.05}}
\end{figure}

\begin{figure}
	\centering
	\includegraphics[width = 0.9\textwidth]{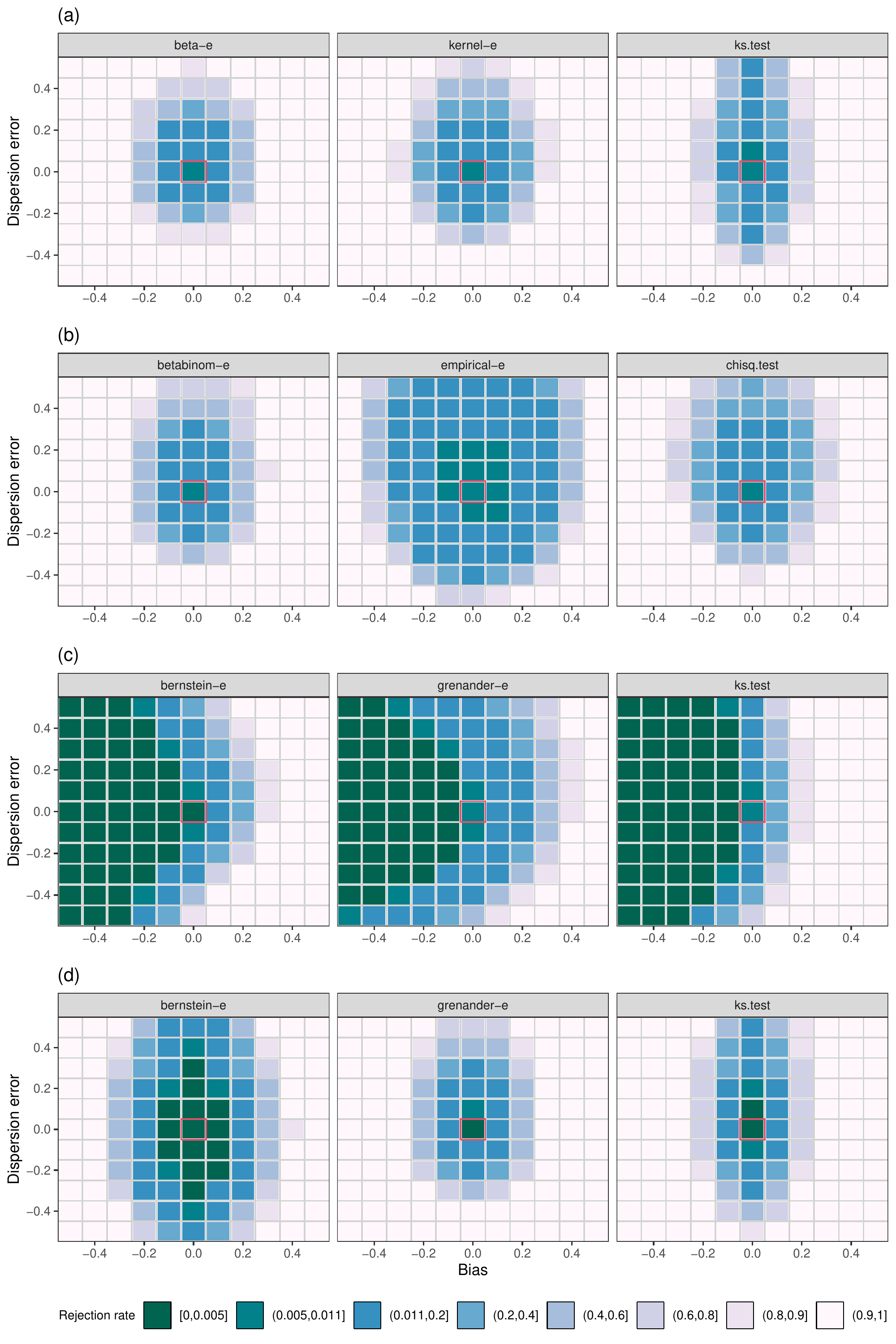}
	\caption{Rejection rates of different tests for (a) the continuous uniform distribution (b) the discrete uniform distribution (c) stochastic dominance (d) calibration of quantile forecasts, at the level $\alpha = 0.01$ with a sample size of $n = 360$, depending on the bias and dispersion error. The red box highlights the rejection rates for bias and dispersion error equal to zero. \label{fig:simulations_power_n_360_alpha_0.01}}
\end{figure}

\begin{figure}
	\centering
	\includegraphics[width = 0.9\textwidth]{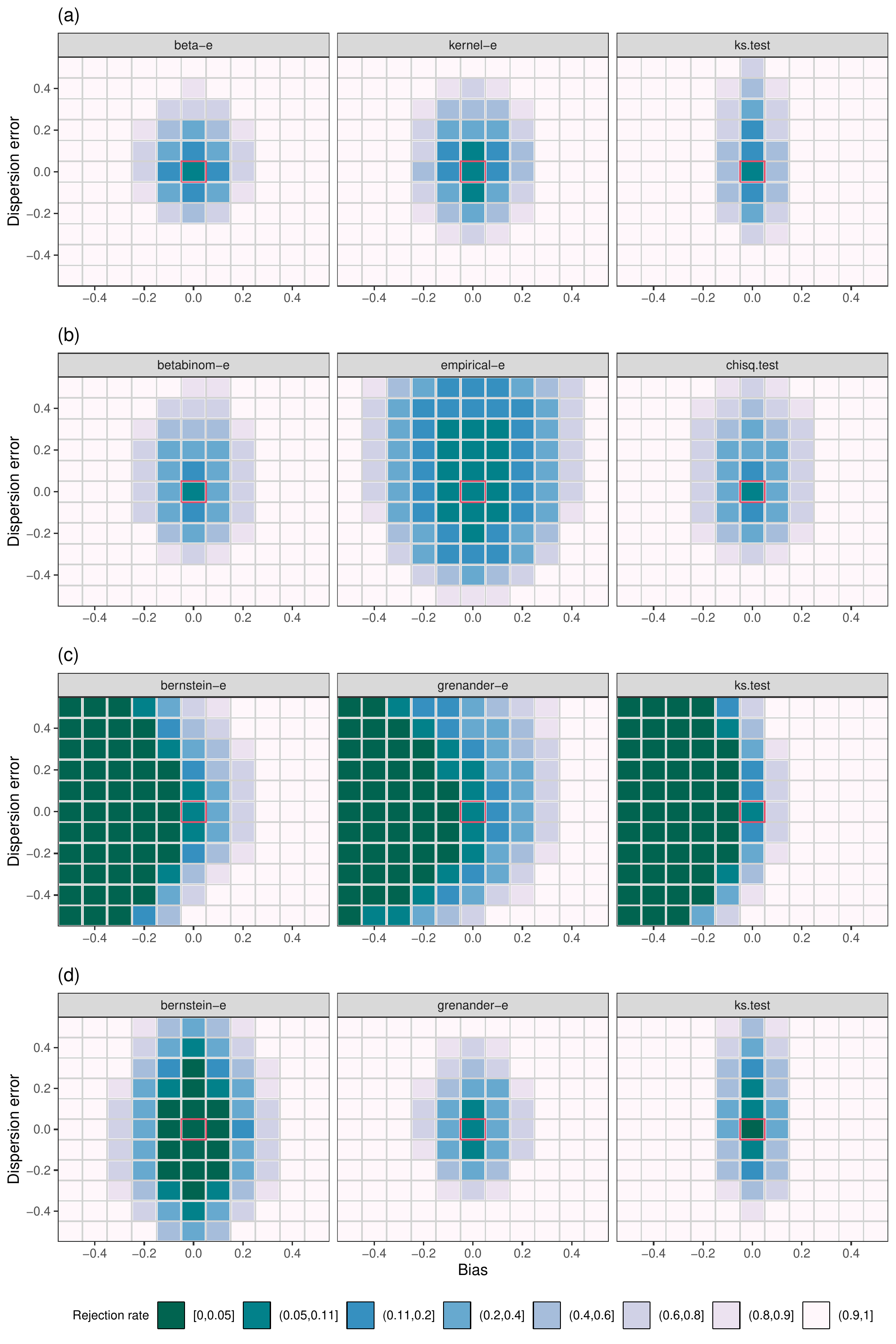}
	\caption{Rejection rates of different tests for (a) the continuous uniform distribution (b) the discrete uniform distribution (c) stochastic dominance (d) calibration of quantile forecasts, at the level $\alpha = 0.1$ with a sample size of $n = 360$, depending on the bias and dispersion error. The red box highlights the rejection rates for bias and dispersion error equal to zero. \label{fig:simulations_power_n_360_alpha_0.1}}
\end{figure}

\begin{figure}
	\centering
	\includegraphics[width = 0.9\textwidth]{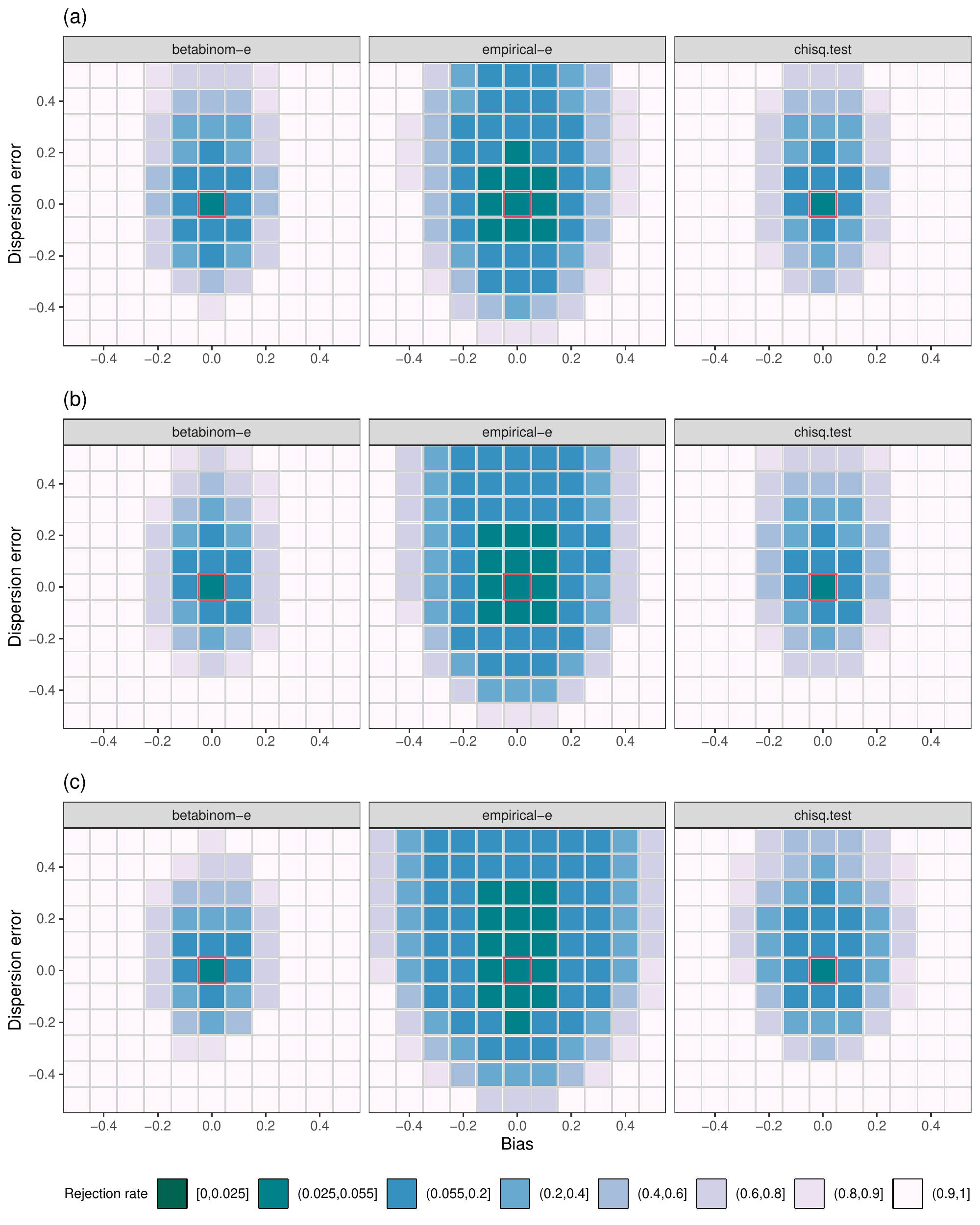}
	\caption{Rejection rates of different tests for uniformity of the rank histogram with ensemble sizes of (a) $m = 10$ (b) $m = 20$ (c) $m = 50$, with $n = 360$ and at the level $\alpha = 0.05$. \label{fig:simulations_power_rhist}}
\end{figure}

\begin{figure}
	\centering
	\includegraphics[width = 0.9\textwidth]{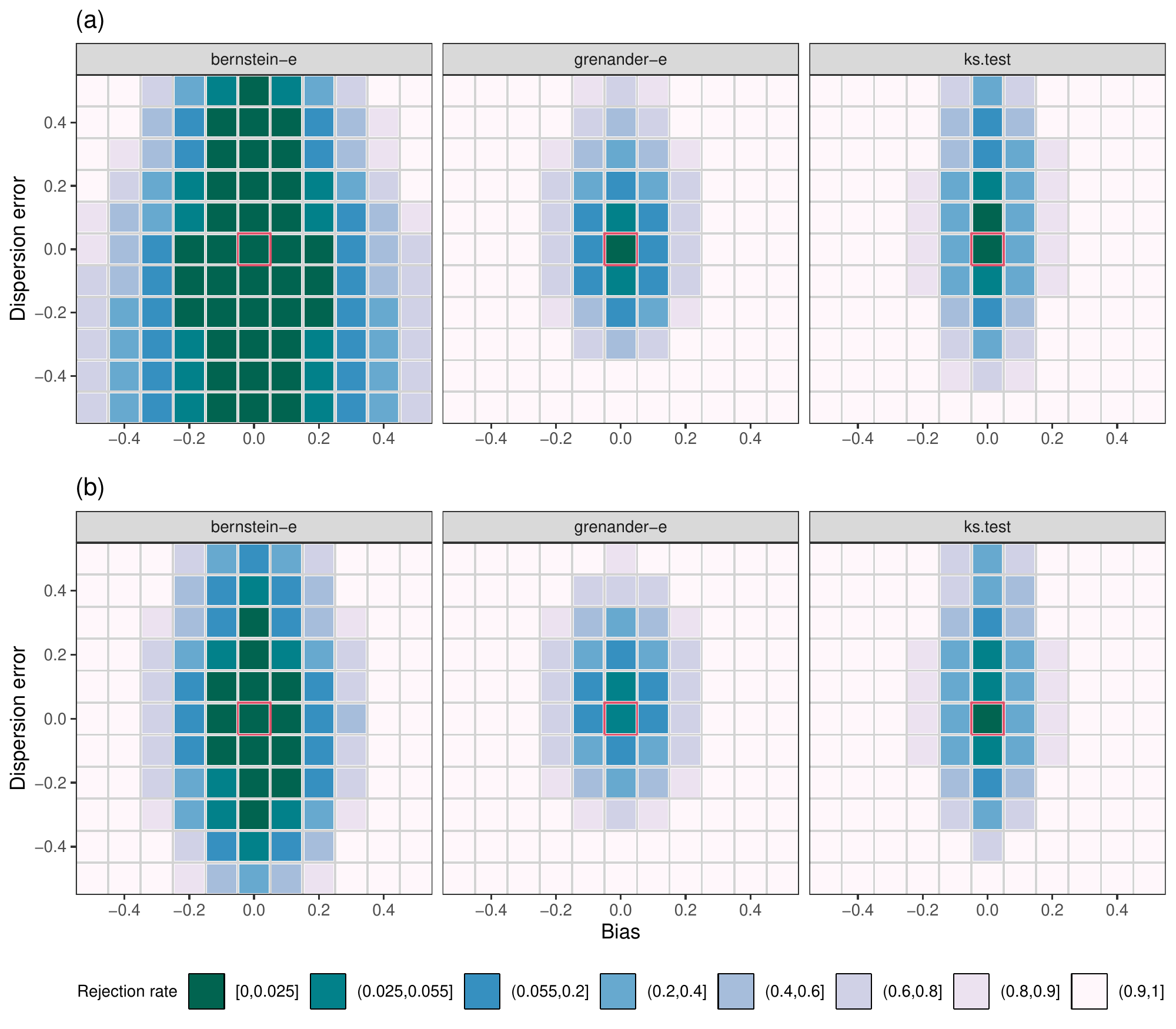}
	\caption{Rejection rates of different tests for calibration of quantile forecasts, with (a) $K = 9$ and (b) $K = 19$ equispaced quantiles, $n = 360$, and $\alpha = 0.05$. \label{fig:simulations_power_qpit}}
\end{figure}

\begin{table}
	\caption{Number of observations (after the first $n_0 = 366$ days where e-values are not computed) until the e-value with the given method (empirical distribution and betabinomial distribution) exceeds the level $10^8$ for the first time, for each weather station, variable, and leadtime. Note that due to the fact that both methods have very similar power, the values often coincide for the two different methods. \label{tab:rhist}}
	\bigskip
	\resizebox{\textwidth}{!}{%
		\footnotesize
		\begin{tabular}{llrrrrrr|llrrrrrr}
			\toprule
			Station ID & Variable & \multicolumn{3}{c}{Empirical distr.} & \multicolumn{3}{c|}{Betabinomial} &
			Station ID & Variable & \multicolumn{3}{c}{Empirical distr.} & \multicolumn{3}{c}{Betabinomial} \\
			\midrule
			& \multicolumn{1}{r}{Leadtime} & $24$ & $48$ & $72$ & $24$ & $48$ & $72$ &
			& \multicolumn{1}{r}{Leadtime} & $24$ & $48$ & $72$ & $24$ & $48$ & $72$\\
			\midrule
			10015 & Precipitation & 19 & 39 & 86 & 19 & 39 & 86 & 10488 & Precipitation & 17 & 73 & 340 & 17 & 73 & 340\\
			& Temperature & 37 & 99 & 141 & 37 & 99 & 141 &  & Temperature & 43 & 137 & 688 & 27 & 137 & 688\\
			& Wind speed & 24 & 416 & 1428 & 27 & 416 & 1428 &  & Wind speed & 34 & 153 & 309 & 34 & 153 & 309\\
			10147 & Precipitation & 21 & 100 & 328 & 21 & 100 & 328 & 10513 & Precipitation & 19 & 71 & 152 & 17 & 71 & 152\\
			& Temperature & 22 & 205 & 550 & 22 & 205 & 550 &  & Temperature & 13 & 109 & 382 & 13 & 109 & 382\\
			
			& Wind speed & 71 & 307 & 1096 & 69 & 307 & 1096 &  & Wind speed & 17 & 57 & 96 & 17 & 57 & 96\\
			10162 & Precipitation & 19 & 128 & 141 & 19 & 128 & 141 & 10554 & Precipitation & 21 & 37 & 140 & 21 & 37 & 140\\
			& Temperature & 29 & 157 & 442 & 31 & 157 & 442 &  & Temperature & 21 & 45 & 151 & 21 & 45 & 151\\
			& Wind speed & 21 & 100 & 271 & 21 & 100 & 271 &  & Wind speed & 18 & 137 & 597 & 18 & 137 & 597\\
			10224 & Precipitation & 31 & 123 & 233 & 31 & 123 & 233 & 10609 & Precipitation & 19 & 49 & 317 & 19 & 49 & 317\\
			
			& Temperature & 36 & 165 & 538 & 36 & 165 & 538 &  & Temperature & 30 & 109 & 243 & 30 & 109 & 243\\
			& Wind speed & 47 & 442 & 1258 & 47 & 442 & 1258 &  & Wind speed & 31 & 96 & 150 & 29 & 96 & 150\\
			10338 & Precipitation & 23 & 66 & 144 & 23 & 66 & 144 & 10637 & Precipitation & 18 & 92 & 172 & 19 & 92 & 172\\
			& Temperature & 18 & 137 & 460 & 18 & 137 & 460 &  & Temperature & 17 & 76 & 388 & 17 & 76 & 388\\
			& Wind speed & 27 & 98 & 621 & 26 & 98 & 621 &  & Wind speed & 69 & 196 & 485 & 60 & 196 & 485\\
			
			10361 & Precipitation & 14 & 55 & 244 & 13 & 55 & 244 & 10655 & Precipitation & 11 & 37 & 320 & 11 & 37 & 320\\
			& Temperature & 15 & 169 & 679 & 15 & 169 & 679 &  & Temperature & 15 & 46 & 420 & 14 & 46 & 420\\
			& Wind speed & 27 & 76 & 99 & 26 & 76 & 99 &  & Wind speed & 44 & 220 & 629 & 48 & 220 & 629\\
			10379 & Precipitation & 14 & 105 & 333 & 14 & 105 & 333 & 10729 & Precipitation & 25 & 43 & 316 & 21 & 43 & 316\\
			& Temperature & 20 & 111 & 418 & 19 & 111 & 418 &  & Temperature & 17 & 73 & 174 & 17 & 73 & 174\\
			
			& Wind speed & 24 & 286 & 687 & 22 & 286 & 687 &  & Wind speed & 21 & 140 & 409 & 19 & 140 & 409\\
			10382 & Precipitation & 10 & 42 & 111 & 10 & 42 & 111 & 10738 & Precipitation & 26 & 77 & 129 & 24 & 77 & 129\\
			& Temperature & 16 & 165 & 693 & 16 & 165 & 693 &  & Temperature & 34 & 75 & 675 & 33 & 75 & 675\\
			& Wind speed & 22 & 219 & 366 & 21 & 219 & 366 &  & Wind speed & 26 & 139 & 143 & 18 & 139 & 143\\
			10400 & Precipitation & 19 & 72 & 76 & 18 & 72 & 76 & 10763 & Precipitation & 12 & 59 & 105 & 12 & 59 & 105\\
			
			& Temperature & 14 & 160 & 498 & 16 & 160 & 498 &  & Temperature & 22 & 86 & 398 & 22 & 86 & 398\\
			& Wind speed & 25 & 140 & 690 & 25 & 140 & 690 &  & Wind speed & 37 & 104 & 264 & 27 & 104 & 264\\
			10444 & Precipitation & 33 & 79 & 185 & 33 & 79 & 185 & 10776 & Precipitation & 21 & 93 & 472 & 21 & 93 & 472\\
			& Temperature & 38 & 302 & 476 & 38 & 302 & 476 &  & Temperature & 21 & 48 & 393 & 20 & 48 & 393\\
			& Wind speed & 42 & 149 & 429 & 42 & 149 & 429 &  & Wind speed & 43 & 153 & 287 & 44 & 153 & 287\\
			
			10469 & Precipitation & 20 & 71 & 340 & 19 & 71 & 340 & 10852 & Precipitation & 21 & 33 & 105 & 21 & 33 & 105\\
			& Temperature & 18 & 44 & 73 & 17 & 44 & 73 &  & Temperature & 19 & 45 & 67 & 19 & 45 & 67\\
			& Wind speed & 36 & 190 & 707 & 32 & 190 & 707 &  & Wind speed & 47 & 392 & 1531 & 45 & 392 & 1531\\
			\bottomrule
		\end{tabular}
	}
\end{table}

\begin{figure}
	\centering
	\includegraphics[width = 0.9\textwidth]{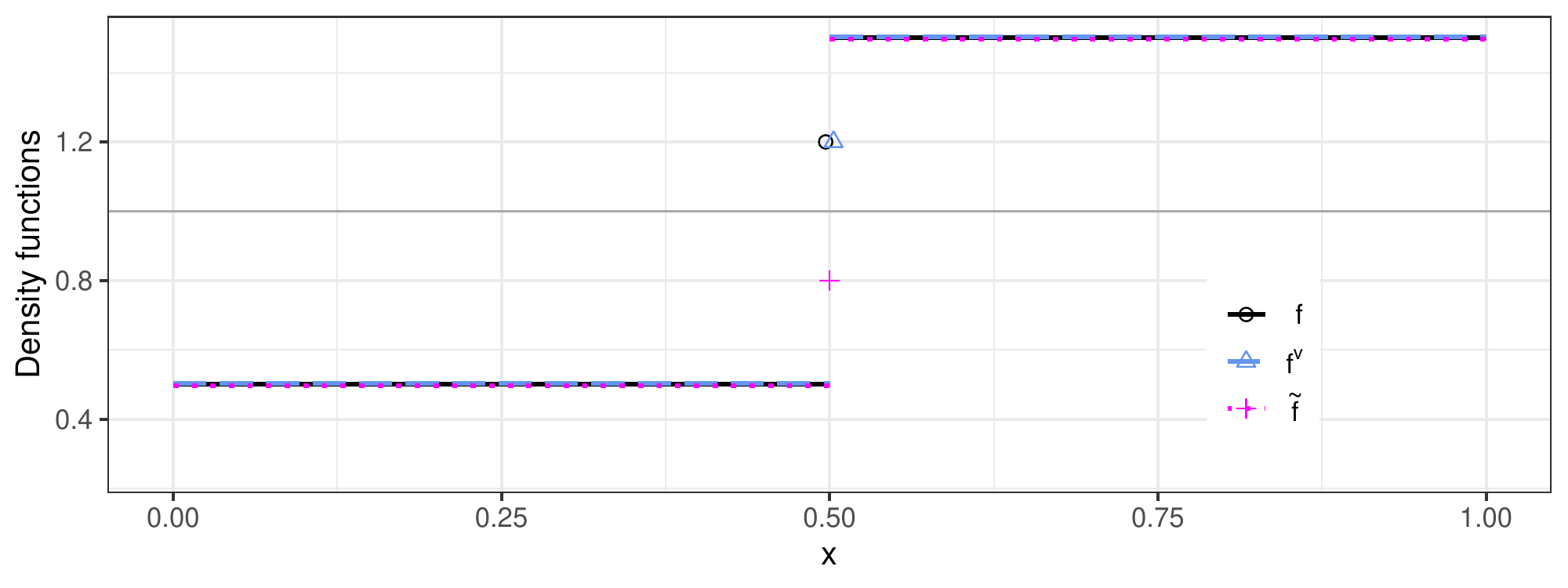}
	\caption{Illustration of the functions $f$, $\tilde{f}$, and $\check{f}$ in the special case in the proof of Proposition 3.2. Here $f(x) = \tilde{f}(x) = \check{f}(x) = 0.5$ for $x < 0.5$ and $f(x) = \tilde{f}(x) = \check{f}(x) = 1.5$ for $x > 0.5$. At $a = 0.5$ we have $\tilde{f}(a) = 0.8 < 1.2 = f(a)$, and we set $\check{f}(a) = f(a) = 1.2$. \label{fig:proof_illustration}}
\end{figure}
\end{document}